\documentclass{article}
\usepackage[utf8]{inputenc}
\usepackage{pig}
\usepackage{graphicx}
\usepackage{fullpage}
\usepackage{enumerate}
\usepackage{subcaption}
\usepackage{placeins}

\newcommand\triv{\mathrm{triv}}

\title{Message-passing algorithms for synchronization problems \\ over compact groups}

\usepackage{authblk}
\author[1]{Amelia Perry%
\footnote{The first two authors contributed equally.}%
\thanks{Email: {\tt ameliaperry@mit.edu}. This work is supported in part by NSF CAREER Award CCF-1453261 and a grant from the MIT NEC Corporation.}%
}
\newcommand{\firstauthmark}{\footnotemark[1]}
\author[1]{Alexander S.\ Wein%
\protect\firstauthmark%
\thanks{Email: {\tt awein@mit.edu}. This research was conducted with Government support under and awarded by DoD, Air Force Office of Scientific Research, National Defense Science and Engineering Graduate (NDSEG) Fellowship, 32 CFR 168a.}%
}
\author[3]{Afonso S.\ Bandeira%
\thanks{Email: {\tt bandeira@cims.nyu.edu}. A.S.B.\ was supported by NSF Grant DMS-1317308.
 Part of this work was done while A.S.B.\ was with the Department of Mathematics at the Massachusetts Institute of Technology.}%
}
\author[1,2]{Ankur Moitra%
\thanks{Email: {\tt moitra@mit.edu}. This work is supported in part by NSF CAREER Award CCF-1453261, a grant from the MIT NEC Corporation and a Google Faculty Research Award.}%
}
\affil[1]{Department of Mathematics, Massachusetts Institute of Technology}
\affil[2]{Computer Science and Artificial Intelligence Lab, Massachusetts Institute of Technology}
\affil[3]{Department of Mathematics and Center for Data Science, Courant Institute of Mathematical Sciences, New York University, NY, USA}

\begin{document}
\maketitle

\begin{abstract}
Various alignment problems arising in cryo-electron microscopy, community detection, time synchronization, computer vision, and other fields fall into a common framework of synchronization problems over compact groups such as $\mathbb{Z}/L$, $U(1)$, or $SO(3)$. The goal of such problems is to estimate an unknown vector of group elements given noisy relative observations. We present an efficient iterative algorithm to solve a large class of these problems, allowing for any compact group, with measurements on multiple `frequency channels' (Fourier modes, or more generally, irreducible representations of the group). Our algorithm is a highly efficient iterative method following the blueprint of approximate message passing (AMP), which has recently arisen as a central technique for inference problems such as structured low-rank estimation and compressed sensing. We augment the standard ideas of AMP with ideas from representation theory so that the algorithm can work with distributions over compact groups. Using standard but non-rigorous methods from statistical physics we analyze the behavior of our algorithm on a Gaussian noise model, identifying phases where the problem is easy, (computationally) hard, and (statistically) impossible. In particular, such evidence predicts that our algorithm is information-theoretically optimal in many cases, and that the remaining cases show evidence of statistical-to-computational gaps.
\end{abstract}

\section{Introduction}
Among the most common data problems in the sciences is that of recovering low-rank signal present in a noisy matrix. The standard tool for such problems is principal component analysis (PCA), which estimates the signal by the top eigenvectors. One example out of many is in macroeconomics, where large, noisy correlation matrices reveal useful volatility and yield predictions in their top eigenvectors \cite{macro1,macro2}. However, many particular applications involve extra structure such as sparsity in the signal, and this structure is ignored by conventional PCA, leading to sub-optimal estimates. Thus, a major topic of recent interest in machine learning is to devise efficient algorithms for sparse PCA \cite{amini-wainwright,berthet-rigollet-optimal}, non-negative PCA \cite{mr-nonnegative}, general Bayesian PCA with a prior \cite{bishop-bayesian}, and other variants. These problems pose a major computational challenge. While significant advances have appeared, it is also expected that there are fundamental gaps between what is statistically possible and what can be done efficiently \cite{br-hardness,mmse-low-rank,phase-sparse,ma-wu}, and thus carried out in practice on very large datasets that are now prevalent.

A number of low-rank recovery problems involve a significant amount of symmetry and group structure. A compelling example is the orientation problem in cryo-electron microscopy (cryo-EM), where one is given many noisy 2D images of copies of an unknown molecule, each in different unknown 3D orientations. The goal is to estimate the orientations, in order to assemble the images into an estimate of the molecule structure \cite{singer-shkolnisky}. Thus, one is tasked with learning elements $g_u$ of $SO(3)$, one for each image $u$, based on some loss function derived from the observed images. Moreover, this loss function has a symmetry: it depends only on the relative alignments $g_u g_v^{-1}$, as there is no {\it a priori} reference frame. One previous approach to this problem, due to \cite{coifman2010reference,singer-shkolnisky}, produces a matrix of pairwise image comparisons, and then attempts to extract the rotations $g_u$ from the top eigenvectors of this matrix. However, it is reasonable to imagine that this approach could be significantly sub-optimal: PCA does not exploit the significant group structure of the signal.

Many problems with similar patterns of group symmetry have been previously studied under the general heading of \emph{synchronization problems}. In general, a synchronization problem asks us to recover a vector of group elements given noisy pairwise measurements of the relative group elements $g_u g_v^{-1}$, and PCA-based `spectral methods' are among the most common techniques used for such problems. Singer \cite{singer2011angular} introduced a PCA approach for angular synchronization, a 2D analogue of the problem above with symmetry over $SO(2)$, in which one estimates the orientations of noisy, randomly rotated copies of an unknown image. Cucuringu et al.\ \cite{cucuringu2012sensor} applied a similar approach to a sensor localization problem with synchronization structure over the Euclidean group $\mathrm{Euc}(2)$. The problem of detecting two subcommunities in a random graph can be viewed as synchronization over $\ZZ/2$ \cite{afonso-thesis}, and spectral methods have a long history of use in such community detection and minimum cut problems (e.g. \cite{mcsherry}). Further instances of synchronization appear in time synchronization in networks \cite{giridhar2006distributed}, computer vision \cite{agrawal2006range}, optics \cite{rubinstein2001reconstruction}, and alignment in signals processing \cite{bandeira2014multireference}.

Further work on synchronization has focused on ways to better exploit the group structure. One method used in practice for cryo-EM and related problems is \emph{alternating minimization}, which alternates between estimating the rotations by aligning the images with a previous guess of the molecule structure, and then estimating the molecule structure from the images using these rotations. This method only appears to succeed given a strong initial guess of the molecule structure, and then it is unclear whether the final estimate mainly reflects the observations or simply the initial guess, leading to a problem of model bias; see e.g.\ \cite{cohen-model-bias}. In this paper we are interested in {\it de novo} estimation without a substantial initial guess, steering clear of this pitfall.

Convex relaxations of maximum likelihood estimators have shown promise, such semidefinite relaxation for angular synchronization introduced by Singer \cite{singer2011angular} and proven tight by Bandeira et al.\ \cite{bandeira2014tightness}, or in semidefinite programs for community detection and correlation clustering problems that constitute $\ZZ/2$-synchronization (e.g.\ \cite{gw,abh,hwx,abbe2014decoding,ms-sdp}). Indeed, a very general semidefinite relaxation for synchronization problems was introduced by Bandeira et al.\ \cite{nug}, but its performance remains unclear even empirically: while this convex program can be solved in polynomial time, it is large enough to make experiments or application difficult.

An alternate approach is an iterative method due to Boumal \cite{boumal}, for the following Gaussian model of angular synchronization: one wants to estimate a vector $x \in \CC^n$ whose entries are unit-norm complex numbers (standing in for 2D rotations), given a matrix
$$ Y = \frac{\lambda}{n} x x^* + \frac{1}{\sqrt{n}} W, $$
where $\lambda$ is a signal-to-noise ratio (SNR) parameter, $W$ is a GUE matrix (independent complex Gaussians up to Hermitian symmetry), and $*$ denotes conjugate transpose. One could perform ordinary PCA by initializing with a small random guess $v$ and repeatedly assigning $v \leftarrow Y v$; this is the method of power iteration.
Instead, Boumal proposes\footnote{Projected power methods have appeared earlier in the literature, for instance, in \cite{mr-nonnegative}.} to iterate $v \leftarrow f(Y v)$, where $f$ divides each entry by its norm, thus projecting to the unit circle. This method is highly efficient, and is moreover observed to produce a better estimate than PCA once the signal-to-noise parameter $\lambda$ is sufficiently large. However, while PCA produces a nontrivial estimate for all $\lambda > 1$ (e.g.\ \cite{fp,nong-eigv1}), this projected power method does not appear to produce a meaningful estimate until $\lambda$ is somewhat larger. (In fact, a heuristic analysis similar to Section~\ref{sec:se} suggests that $\lambda > 2/\sqrt{\pi} \approx 1.128$ is required.) This behavior suggests that some iterative method combining the best features of PCA and the projected power method might outperform both statistically, while remaining very efficient. More importantly, we are motivated to find analogous iterative methods for groups other than $U(1)$, and for more complicated observation models.

One such observation model is as follows. Instead of observing only the matrix $Y$ as above, suppose we are given matrices corresponding to different Fourier modes:
$$ Y_1 = \frac{\lambda_1}{n} x x^* + \frac{1}{\sqrt{n}} W_1, $$
$$ Y_2 = \frac{\lambda_2}{n} x^2 (x^2)^* + \frac{1}{\sqrt{n}} W_2, $$
$$ \vdots $$
$$ Y_K = \frac{\lambda_K}{n} x^K (x^K)^* + \frac{1}{\sqrt{n}} W_K, $$
where $x^k$ denotes entrywise power, and $W_k$ are independent GUE matrices. With a PCA-based approach, it is not clear how to effectively couple the information from these matrices to give a substantially better estimate than could be derived from only one. The ``non-unique games'' semidefinite program of \cite{nug} is able to use data from an observation model such as this, but it is not yet empirically clear how it performs. Can we hope for some very efficient iterative algorithm to strongly leverage data from multiple `frequencies' or `channels' such as this?

This question has more than abstract relevance: due to Fourier theory, a very large class of measurement models for $U(1)$-synchronization decomposes into matrix-based observations on different frequencies, in a manner resembling the model above. Moreover, an analogous statement holds over $SO(3)$, and over other compact groups, with Fourier theory replaced by the noncommutative setting of representation theory. Thus, the aforementioned spectral approach to cryo-EM applies PCA to only the lowest frequency part of the observations; an algorithm that can use all frequencies effectively might demonstrate dramatically improved statistical performance on cryo-EM.

In this paper we present an iterative algorithm to meet the challenges above. Our algorithm aims to solve a general formulation of the synchronization problem: it can apply to multiple-frequency problems for a large class of observation models, with symmetry over any compact group. Our approach is statistically powerful, empirically providing a better estimate than both PCA and the projected power method on $U(1)$-synchronization, and leveraging multiple frequencies to give several orders of magnitude improvement in estimation error in experiments (see Figures~\ref{fig:u1-many} and~\ref{fig:u1-many-log}). Indeed, we conjecture based on ideas from statistical physics that in many regimes our algorithm is statistically optimal, providing a minimum mean square error (MMSE) estimator asymptotically as the matrix dimensions become infinite (see Section~\ref{sec:gaps}). Finally, our approach is highly efficient, with each iteration taking time linear in the (matrix) input, and with roughly 15 iterations sufficing for convergence in experiments.

Our algorithm follows the framework of \emph{approximate message passing} (AMP), based on belief propagation on graphical models \cite{pearl} and the related cavity method in statistical physics \cite{mezard-parisi-virasoro}. Following a general blueprint, AMP algorithms have previously been derived and analyzed for compressed sensing \cite{amp-cs,amp-mot,bm,jm}, sparse PCA \cite{dm-sparse-pca}, non-negative PCA \cite{mr-nonnegative}, cone-constrained PCA \cite{dmr-cone}, planted clique \cite{dm-clique} and general structured PCA \cite{RF-amp}. In fact, AMP has already been derived for $\ZZ/2$-synchronization under a Gaussian observation model \cite{dam}, and our algorithm will generalize this one to all compact groups. A striking feature of AMP is that its asymptotic performance can be captured exactly by a particular fixed-point equation called \emph{state evolution}, which has enabled the rigorous understanding of its performance on some problems \cite{bm,jm}. AMP is provably statistically optimal in many cases, including Gaussian $\ZZ/2$ synchronization (modulo a technicality whereby the proof supposes a small warm-start) \cite{dam}.

AMP algorithms frequently take a form similar to the projected power method of Boumal described above, alternating between a matrix--vector product with the observations and an entrywise nonlinear transformation, together with an extra `Onsager' correction term. In the case of $\ZZ/2$- or $U(1)$-synchronization, we will see that the AMP derivation reproduces Boumal's algorithm, except with the projection onto the unit circle replaced by a soft, sigmoid-shaped projection function to the unit disk (see Figure~\ref{fig:tanh}), with the magnitude maintaining a quantitative measure of confidence. Integrating the usual AMP blueprint with the representation theory of compact groups, we obtain a broad generalization of this method, to synchronization problems with multiple frequencies and noncommutative groups such as $SO(3)$. In full generality, the nonlinear transformation has a simple interpretation through representation theory and the exponential function.

One drawback of our approach is that although we allow for a very general observation model, we do insist that the noise on each pairwise measurement is independent. This fails to capture certain synchronization models such as multireference alignment \cite{bandeira2014multireference} and cryo-EM that have noise on each group element rather than on each pair. For instance, the noise in cryo-EM occurs on each image rather than independently on each pairwise comparison. Adapting AMP to these more general models is left for future work.

This paper is organized as follows. We begin in Section~\ref{sec:intuition} with an outline of our methods in the simplified cases of synchronization over $\ZZ/2$ and $U(1)$, motivating our approach from a detailed discussion of prior work and its shortfalls. In Section~\ref{sec:alg} we provide our general algorithm and the general problem model for which it is designed. Several experiments on this Gaussian model and other models are presented in Section~\ref{sec:experiments}, demonstrating strong empirical performance. We then offer two separate derivations of our AMP algorithm: in Section~\ref{sec:bp-deriv}, we derive our algorithm as a simplification of belief propagation, and then in Section~\ref{sec:mmse-deriv} we give an alternative self-contained derivation of the nonlinear update step and use this to provide a non-rigorous analysis of AMP (based on standard assumptions from statistical physics). In particular, we derive the state evolution equations that govern the behavior of AMP, and use these to identify the threshold above which AMP achieves non-trivial reconstruction. Namely, we see that AMP has the same threshold as PCA (requiring the SNR $\lambda$ to exceed $1$ on at least one frequency), but AMP achieves better recovery error above the threshold. In Section~\ref{sec:se-correct} we argue for the correctness of the above non-rigorous analysis, providing both numerical and mathematical evidence. It is known that inefficient estimators can beat the $\lambda = 1$ threshold \cite{pwbm-contiguity} but we conjecture that no efficient algorithm is able to break this barrier, thus concluding in Section~\ref{sec:gaps} with an exploration of statistical-to-computational gaps that we expect to exist in synchronization problems, driven by ideas from statistical physics.


\section{\texorpdfstring{Intuition: iterative methods for $\ZZ/2$ and $U(1)$ synchronization}{Intuition: iterative methods for Z/2 and U(1) synchronization}}\label{sec:intuition}

We begin with a discussion of synchronization methods over the cyclic group $\ZZ/2$ and the group of unit complex numbers (or 2D rotations) $U(1)$. These examples will suffice to establish intuition and describe much of the novelty of our approach, while avoiding the conceptual and notational complication of representation theory present in the general case. Sections~\ref{sec:z2}, \ref{sec:bp-amp}, and some of~\ref{sec:amp-u1-one} discuss prior work on these problems in more depth, while Sections~\ref{sec:amp-u1-one} and~\ref{sec:amp-u1-mult} develop a special case of our algorithm.

\subsection{\texorpdfstring{$\ZZ/2$ synchronization}{Z/2 synchronization}}\label{sec:z2}
The problem of \emph{Gaussian $\ZZ/2$ synchronization} is to estimate a uniformly drawn signal $x \in \{\pm 1\}^n$ given the matrix
$$ Y = \frac{\lambda}{n} x x^\top + \frac{1}{\sqrt n} W, $$
where $W$ is a symmetric matrix whose entries are distributed independently (up to symmetry) as $\cN(0,1)$, and $\lambda > 0$ is a signal-to-noise parameter. With this scaling, the signal and noise are of comparable size in spectral norm; we can not hope to recover $x$ exactly, but we can hope to produce an estimate $\hat x \in \{\pm 1\}$ that is correlated nontrivially with $x$, i.e. there exists $\eps > 0$ (not depending on $n$) such that $\frac{1}{n^2} \langle x,\hat x \rangle^2 > \eps$ with probability $1-o(1)$ as $n \to \infty$. As $x x^\top = (-x) (-x)^\top$, we can only hope to estimate $x$ up to sign; thus we aim to achieve a large value of $\langle x, \hat x \rangle^2$. We now review three algorithmic methods for this problem.

\paragraph{Spectral methods.} With the scaling above, the spectral norm of the signal $\frac{\lambda}{n} x x^\top$ is $\lambda$, while that of the noise $\frac{1}{\sqrt{n}} W$ is $2$. By taking the top eigenvector of $Y$, $x$ may be estimated with significant correlation provided that $\lambda$ is a large enough constant.

Specifically, the generative model for $Y$ above is a special case of the \emph{spiked Wigner model}, and the eigenvalues and eigenvectors of such spiked models are among the main objects of study in random matrix theory. When $\lambda > 1$, the (unit norm) top eigenvector $v_{\max}(Y)$ correlates nontrivially with $x$; more specifically, as $n \to \infty$, we have $\frac{1}{n} \langle x, v_{\max}(Y) \rangle^2 \to 1 - 1/\lambda^2$ in probability \cite{fp,nong-eigv1}. When $\lambda < 1$, this squared correlation tends to zero; in fact, this is known to be true of all estimators \cite{dam,pwbm-contiguity}, reflecting a sharp statistical phase transition.

Note that a top eigenvector may be computed through \emph{power iteration} as follows: an initial guess $v^0$ is drawn randomly, and then we iteratively compute $v^{(t)} = Y v^{(t-1)}$, rescaling the result as appropriate. Thus each entry is computed as $v^{(t)}_u = \sum_w Y_{uw} v^{(t-1)}_w$; we can imagine that each entry $w$ sends a `message' $Y_{uw} v^{(t-1)}_w$ to each entry $u$ -- the `vote' of entry $w$ as to the identity of entry $u$ -- and then each entry sums the incoming votes to determine its new value. The result has both a sign, reflecting the weighted majority opinion as to whether that entry should ultimately be $+1$ or $-1$, and also a magnitude, reflecting a confidence and serving as the weight in the next iteration. Thus we can envision the spectral method as a basic ``message-passing algorithm''.

While this approach is effective as quantified above, it would seem to suffer from two drawbacks:
\begin{itemize}
\item the spectral method is basis-independent, and thus cannot exploit the entrywise $\pm 1$ structure of the signal;
\item the vertex weights can grow without bound, potentially causing a few vertices to exert undue influence.
\end{itemize}

Indeed, these drawbacks cause major issues in the \emph{stochastic block model}, a variant of the model above with the Gaussian observations replaced by low-probability Bernoulli observations, usually envisioned as the adjacency matrix of a random graph. Here a few sporadically high-degree vertices can dominate the spectral method, causing asymptotically significant losses in the statistical power of this approach.

\paragraph{Projected power iteration.}
Our next stepping-stone toward AMP is the projected power method studied by \cite{boumal,chen-candes}, a variant of power iteration that exploits entry-wise structure. Here each iteration takes the form $v^{(t)} = \sgn(Y v^{(t-1)})$, where the sign function $\sgn : \mathbb{R} \to \{\pm 1\}$ applies entrywise. Thus each iteration is a majority vote that is weighted only by the magnitudes of the entries of $Y$; the weights do not become more unbalanced with further iterations. This algorithm is also basis-dependent in a way that plausibly exploits the $\pm 1$ structure of the entries.

Empirically, this algorithm obtains better correlation with the truth, on average, when $\lambda > 2.4$ approximately; see Figure~\ref{fig:z2-compare}. However, for very noisy models with $1 < \lambda < 2.4$, this method appears weaker than the spectral method. The natural explanation for this weakness is that this projected power method forgets the distinction between a 51\% vote and a 99\% vote, and thus is over-influenced by weak entries. This is particularly problematic for low signal-to-noise ratios $\lambda$, for which 51\% votes are common. In fact, a heuristic analysis similar to Section~\ref{sec:se} suggests that this method does not achieve the correct threshold for $\lambda$, failing to produce nontrivial correlation whenever $\lambda < 2/\sqrt{\pi} \approx 1.128$.

\paragraph{Soft-threshold power iteration.}
A natural next step is to consider iterative algorithms of the form $v^{(t)} = f(Y v^{(t-1)})$, where $f$ applies some function $\RR \to [-1,1]$ entrywise (by abuse of notation, we will also denote the entrywise function by $f$). Instead of the identity function, as in the spectral method, or the sign function, as in the projected power method, we might imagine that some continuous, sigmoid-shaped function performs best, retaining some sense of the confidence of the vote without allowing the resulting weights to grow without bound. It is natural to ask what the optimal function for this purpose is, and whether the resulting weights have any precise meaning.

Given the restriction to the interval $[-1,1]$, one can imagine treating each entry as a sign with confidence in a more precise way, as the expectation of a distribution over $\{\pm 1\}$. At each iteration, each entry $u$ might then obtain the messages $Y_{uw} v^{(t-1)}_w$ from all others, and compute the posterior distribution, summarized as an expectation $v^{(t)}_u$. As one can compute, this corresponds to the choice of transformation $f(t) = \tanh(\lambda t)$ where $\lambda$ is the signal-to-noise parameter from above (see Figure~\ref{fig:tanh}).

\begin{figure}[!ht]\centering
\begin{minipage}{0.8\textwidth}\centering
\includegraphics[width=0.6\linewidth]{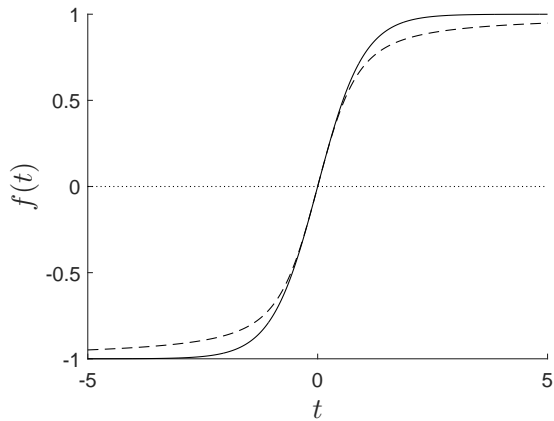}
\caption{Soft threshold functions used by AMP. The solid line is $f(t) = \tanh(t)$, used for $\mathbb{Z}/2$ synchronization. The dashed line is $f(t) = I_1(2t)/I_0(2t)$ (modified Bessel functions of the first kind), used for $U(1)$ synchronization with one frequency.}
\label{fig:tanh}
\end{minipage}
\end{figure}

\subsection{Belief propagation and approximate message passing}\label{sec:bp-amp}

The soft-projection algorithm above may remind the reader of belief propagation, due to \cite{pearl} in the context of inference, and to \cite{mezard-parisi-virasoro} as the \emph{cavity method} in the context of statistical physics. We may envision the problem of estimating $x$ as probabilistic inference over a graphical model. The vertices of the model represent the unknown entries of $x$, and every pair of vertices $u,w$ participates in an edge interaction based on the matrix entry $Y_{uw} = Y_{wu}$. Specifically, it may be computed that the posterior distribution for $x \in \{\pm 1\}^n$ after observing $Y$ is given by
$$ \Pr(x) \propto \prod_{u < w} \exp(\lambda Y_{uw} x_u x_w), $$
which is precisely the factorization property that a graphical model captures.

Given such a model, belief propagation proceeds in a fashion reminiscent of the previous algorithm: each vertex $w$ sends a message to each neighbor $u$ encoding the posterior distribution of $x_u$ based the previous distribution of $x_w$ and the direct interaction $\lambda Y_{uw}$. Each vertex $u$ then consolidates all incoming messages into a new `posterior' distribution on $x_u$ given these messages, computed as if the messages were independent. However, belief propagation introduces a correction to this approach: rather than letting information passed from $w$ to $u$ propagate back to $w$ on the next iteration, belief propagation is designed to pass information along only those paths that do not immediately backtrack. Specifically, at each iteration, the message from $w$ to $u$ is based on only the synthesis of messages from all vertices except $u$ from the previous iteration.

This algorithm differs from the iterative methods presented above, both in this non-backtracking behavior, and in the fact that the transformation from the distribution at $w$ to the message $w \to u$ is not necessarily linear (as in the multiplication $Y_{uw} v_w$ above). Both of these differences can be reduced by passing to the framework of \emph{approximate message passing} \cite{amp-cs}, which simplifies belief propagation in dense models with weak interactions, through the following two observations (inspired by \cite{tap} in the theory of spin glasses):
\begin{itemize}
\item As the interaction $\lambda Y_{uw}$ is small, scaling as $O(1/\sqrt{n})$ as $n \to \infty$, we may pass to an expansion in small $Y_{uw}$ when computing the message $w \to u$ from the mean $w$. In this example, we find that the message $w \to u$ should be Rademacher with mean $\lambda Y_{uw} v_w^{(t-1)} + O(Y_{uw}^2)$, where $v_w^{(t-1)}$ is the mean of the distribution for $x_w$ in the previous iteration. This linear expansion ensures that the main message-passing step can be expressed as a matrix--vector product.
\item Rather than explicitly computing non-backtracking messages, which is computationally more involved, we may propagate the more na\"ive backtracking messages and then subtract the bias due to this simplification, which concentrates well. This correction term is called an \emph{Onsager correction}. If vertex $w$ passes messages to all neighbors based on its belief at iteration $t-2$, and then all of these neighbors send return messages based on their new beliefs at time $t-1$, then when updating the belief for vertex $w$ at time $t$, one can explicitly subtract off the `reflected' influence of the previous belief at time $t-2$. It turns out that this is the only correction necessary: all other error contributions (e.g. 3-cycles) are $o(1)$.
\end{itemize}

Following these simplifications, one can arrive at an \emph{approximate message passing} (AMP) algorithm for $\ZZ/2$-synchronization:
\begin{algorithm}[AMP for $\ZZ/2$ synchronization \cite{dam}]---
\begin{itemize}
\item Initialize $v^0$ to small (close to zero) random values in $[-1,1]$.
\item Iterate for $1 \leq t \leq T$:
\begin{itemize}
\item Set $c^{(t)} = \lambda Y v^{(t-1)} - \lambda^2 (1-\langle (v^{(t-1)})^2 \rangle) v^{(t-2)}$, the Onsager-corrected sum of incoming messages.
\item Set $v^{(t)}_u = \tanh(c^{(t)}_u)$ for each vertex $u$, the new estimated posterior mean.
\end{itemize}
\item Return $\hat x = v^{(T)}$ (or the approximate MAP estimate $\hat x = \sgn(v^{(T)})$ if a proper estimate in $\{\pm 1\}^n$ is desired).
\end{itemize}
\end{algorithm}
\noindent Here $\langle (v^{(t-1)})^2 \rangle$ denotes the average of the squared entries of $v^{(t-1)}$. Detailed derivations of this algorithm appear in Sections~\ref{sec:bp-deriv} and~\ref{sec:mmse-deriv} in much higher generality.

In the setting of $\ZZ/2$ synchronization, an algorithm equivalent to the above approach appears in \cite{dam}, where a statistical optimality property is proven: if AMP is warm-started with a state $v^0$ with nontrivial correlation with the truth, then it converges to an estimate of $x$ that achieves minimum mean-squared error (MMSE) asymptotically as $n \to \infty$. The warm-start requirement is technical and likely removable: if AMP is initialized to small randomness, with trivial correlation $O(1/\sqrt{n})$ with the truth, then its early iterations resemble PCA and should produce nontrivial correlation in $O(\log n)$ iterations. The statistical strength of AMP is confirmed empirically, as it appears to produce a better estimate than either PCA or the projected power method, for every $\lambda > 1$; see Figure~\ref{fig:z2-compare}.

\subsection{\texorpdfstring{AMP for Gaussian $U(1)$ synchronization with one frequency}{AMP for Gaussian U(1) synchronization with one frequency}}\label{sec:amp-u1-one}
As a first step toward higher generality, consider the following Gaussian synchronization model over the unit complex numbers $U(1)$. The goal is to estimate a uniformly drawn signal $x \in U(1)$ given the matrix
$$ Y = \frac{\lambda}{n} x x^* + \frac{1}{\sqrt n} W, $$
where $W$ is a Hermitian matrix whose entries are distributed independently as $\CC\cN(0,1)$, the complex normal distribution given by $\cN(0,1/2) + \cN(0,1/2)i$, and where $\lambda > 0$ is a signal-to-noise parameter. As $x x^*$ is invariant under a global phase shift of $x$, we can only hope to estimate up to the same ambiguity, and so we would like an estimator $\hat x$ that maximizes $|\langle x, \hat x\rangle|^2$, where the inner product is conjugated in the second variable. Many of the previously discussed iterative techniques adapt to this new case.

\paragraph{Spectral methods.} The same analysis of the spectral method holds in this case; thus when $\lambda > 1$, the top eigenvector achieves nontrivial correlation with $x$, while for $\lambda < 1$, the spectral method fails and nontrivial estimation is provably impossible \cite{pwbm-contiguity}.

\paragraph{Projected power method.} After each matrix--vector product, we can project $v^{(t)}$ entrywise onto the unit circle, preserving the phase of each entry while setting the magnitude to $1$. This algorithm is analyzed in \cite{boumal} in a lower-noise setting, where it is shown to converge to the maximum likelihood estimator.

\paragraph{Soft-threshold power method.} One might imagine applying some entrywise function after each matrix--vector product, which preserves the phase of each entry while mapping the magnitude to $[0,1]$. Thus the vector entries $v_u$ live in the unit disk, the convex hull of the unit circle; these might be envisioned as estimates of the posterior expectation of $x_u$. 

\paragraph{Belief propagation \& AMP.} Belief propagation is somewhat problematic in this setting: all messages should express a distribution over $U(1)$, and it is not \emph{a priori} clear how this should be expressed in finite space. However, under the simplifications of approximate message passing, the linearity of the message-passing stage enables a small summary of this distribution to suffice: we need only store the expectation of each distribution, a single value in the unit disk. Approximate message passing takes the following form:
\begin{algorithm}[AMP for $U(1)$ synchronization with one frequency]---
\begin{itemize}
\item Initialize $v^0$ to small random values in the unit disk $\conv(U(1))$.
\item Iterate for $1 \leq t \leq T$:
\begin{itemize}
\item Set $c^{(t)} = \lambda Y v^{(t-1)} - \lambda^2 (1-\langle |v^{(t-1)}|^2 \rangle) v^{(t-2)}$, the Onsager-corrected sum of incoming messages.
\item Set $v^{(t)}_u = f(c^{(t)}_u)$ for each vertex $u$, the new estimated posterior mean. Here $f$ applies the function $f(t) = I_1(2t) / I_0(2t)$ to the magnitude, leaving the phase unchanged.
\end{itemize}
\item Return $\hat x = v^{(T)}$ (or the approximate MAP estimate $\hat x = \phase(v^{(T)})$ if a proper estimate in $U(1)^n$ is desired).
\end{itemize}
\end{algorithm}
\noindent Here $I_k$ denotes the modified Bessel functions of the first kind. The function $f$ is depicted in Figure~\ref{fig:tanh}. Detailed derivations of this algorithm appear in Sections~\ref{sec:bp-deriv} and~\ref{sec:mmse-deriv} in much higher generality.

\subsection{\texorpdfstring{AMP for Gaussian $U(1)$ synchronization with multiple frequencies}{AMP for Gaussian U(1) synchronization with multiple frequencies}}\label{sec:amp-u1-mult}
Consider now the following more elaborate synchronization problem. The goal is to estimate a spike $x \in U(1)^n$ from the observations
$$ Y_1 = \frac{\lambda_1}{n} x x^* + \frac{1}{\sqrt{n}} W_1, $$
$$ Y_2 = \frac{\lambda_2}{n} x^2 (x^2)^* + \frac{1}{\sqrt{n}} W_2, $$
$$ \vdots $$
$$ Y_K = \frac{\lambda_K}{n} x^K (x^K)^* + \frac{1}{\sqrt{n}} W_K, $$
where the $W_k$ are independent Hermitian matrices whose entries are distributed independently (up to symmetry) as $\CC\cN(0,1)$, the $\lambda_k > 0$ are signal-to-noise parameters, and $x^k$ denotes the entrywise $k$th power of $x$.

Thus we are given $K$ independent noisy matrix-valued observations of $x$; we can imagine these observations as targeting different \emph{frequencies} or Fourier modes. Given two independent draws of $\lambda \, x x^* / n + W/\sqrt{n}$ as in the previous section, the spectral method applied to their average will produce a nontrivial estimate of $x$ as soon as $\lambda > 1/\sqrt{2}$. However, under the multiple frequencies model above, with $K=2$ and $\lambda_1 = \lambda_2 = \lambda$, nontrivial estimation is provably impossible for $\lambda < 0.937$ \cite{pwbm-contiguity}; we present non-rigorous evidence in Section~\ref{sec:gaps} that the true statistical threshold should in fact remain $\lambda = 1$. Thus the multiple frequencies model would seem to confound attempts to exploit the multiple observations together. However, we will discuss how AMP enables us to obtain a much better estimate when $\lambda > 1$ than is possible with one matrix alone.

Let us return to the issue of belief propagation over $U(1)$, and of how to represent distributions. One crude approach might be to discretize $U(1)$ and express the density on a finite subset of points; however, this is somewhat messy (e.g.\ the discretization is not preserved under rotation) and only becomes worse for more elaborate groups such as $SO(3)$ (here one can not even find arbitrarily fine discretizations on which the group acts transitively).

Instead, we could exploit the rich structure of Fourier theory, and express a distribution on $U(1)$ by the Fourier series of its density\footnote{A dense subset of distributions satisfies appropriate continuity assumptions to discuss their densities with respect to uniform measure, a Fourier series, etc., and we will not address these analytic technicalities further.}. Thus, if $\mu_w^{(t)}$ is the belief distribution at vertex $w$ and time $t$, we can express:
$$ \dd[\mu_w^{(t)}]{\theta/2\pi} = \sum_{k \in \ZZ} v_{w,k} e^{i k \theta}, $$
with $v_{w,0} = 1$ and $v_{w,-k} = \bar{v_{w,k}}$. Computing the distributional message $m_{w \to u}$ from $w$ to $u$, we obtain
$$ \dd[m_{w \to u}^{t+1}]{\theta/2\pi} = 1 + \sum_{1 \leq |k| \leq K} \lambda_k (Y_k)_{uw} \; v_{w,k} e^{i k \theta} + O((Y_\bullet)_{uw}^2), $$
where we take $Y_{-k} = \bar{Y_k}$.
As $(Y_k)_{uw}$ is order $1/\sqrt{n}$ in probability, this approximation will be asymptotically accurate. Thus it suffices to represent distributions by the coefficients $v_{w,k}$ with $|k| \leq K$. By conjugate symmetry, the coefficients with $1 \leq k \leq K$ suffice. The sufficiency of this finite description of each belief distribution is a key insight to our approach.

The other crucial observation concerns the remaining BP step of consolidating all incoming messages into a new belief distribution. As each incoming message is a small perturbation of the uniform distribution, the approximation $\log(1+x) \approx x$ allows us to express the log-density of the message distribution:
$$ \log \dd[m_{w \to u}^{t+1}]{\theta/2\pi} = \sum_{1 \leq |k| \leq K} \lambda_k (Y_k)_{uw} \; v_{w,k} e^{i k \theta} + O((Y_\bullet)_{uw}^2). $$
We now add these log-densities to obtain the log-density of the new belief distribution, up to normalization:
$$ \log \dd[\mu_u^{t+1}]{\theta/2\pi} + \mathrm{const.} = \sum_{1 \leq |k| \leq K} \left( \sum_{w \neq u} \lambda_k (Y_k)_{uw} v_{w,k} \right) e^{i k \theta} + O((Y_\bullet)_{uw}^2). $$
We thus obtain the Fourier coefficients of the \emph{log-density} of the new belief from the Fourier coefficients of the \emph{density} of the old belief, by matrix--vector products. Remarkably, this tells us that the correct per-vertex nonlinear transformation to apply at each iteration is the transformation from Fourier coefficients of the log-density to those of the density! (In section~\ref{sec:mmse-deriv} we will see an alternative interpretation of this nonlinear transformation as an MMSE estimator.)

The only constraints on a valid log-density are those of conjugate symmetry on Fourier coefficients; thus log-densities form an entire linear space. By contrast, densities are subject to non-negativity constraints, and form a nontrivial convex body in $\RR^K$. The latter space is the analogue of the unit disk or the interval $[-1,1]$ in the preceding examples, and this transformation from the Fourier series of a function to those of its exponential (together with normalization) forms the analogue of the preceding soft-projection functions.

We thus arrive at an AMP algorithm for the multiple-frequency problem:
\begin{algorithm}[AMP for $U(1)$ synchronization with multiple frequencies]---
\begin{itemize}
\item For each $1 \leq k \leq K$ and each vertex $u$, initialize $v_{u,k}^0$ to small random values in $\CC$.
\item Iterate for $1 \leq t \leq T$:
\begin{enumerate}
\item For each $1 \leq k \leq K$, set $c^{(t)}_k = \lambda_k Y_k v^{(t-1)}_k - \lambda_k^2 (1-\langle (v^{(t-1)}_k)^2 \rangle) v^{(t-2)}_k$, the vector of $k$th Fourier components of the estimated posterior log-densities, with Onsager correction.
\item Compute $v^{(t)}_k$, the vector of $k$th Fourier components of the estimated posterior densities.
\end{enumerate}
\item Return $\hat x = v^{(T)}_1$ (or some rounding if a proper estimate in $U(1)^n$ is desired, or even the entire per-vertex posteriors represented by $c^{(T)}$).
\end{itemize}
\end{algorithm}
\noindent Again, a more detailed derivation can be found in Sections~\ref{sec:bp-deriv} and~\ref{sec:mmse-deriv}.

It is worth emphasizing that only the expansion
$$ \log \dd[\mu_u^{(t)}]{\theta/2\pi} + \mathrm{const.} = 2 \Re \sum_{1 \leq k \leq K} c_{k,w}^{(t)} e^{i k \theta} $$
is an accurate expansion of the estimated vertex posteriors. While this log-density is band-limited, this still allows for the density to be a very spiked, concentrated function, without suffering effects such as the Gibbs phenomenon. By contrast, the finitely many $v$ coefficients that this algorithm computes do not suffice to express the Fourier expansion of the density, and a truncated expansion based on only the computed coefficients might even become negative.

We conclude this section by noting that nothing in our derivation depended crucially on the Gaussian observation model. The choice of model tells us how to propagate beliefs along an edge according to a matrix--vector product, but we could carry this out for a larger class of graphical models. The essential properties of a model, that enables this approach to adapt, are:
\begin{itemize}
\item The model can be expressed as a graphical model with only pairwise interactions:
$$ \Pr(x) \propto \prod_{u < w} \cL_{uw}(x_u,x_w). $$
\item The interaction graph is dense, with all pair potentials individually weak (here, $1 + O(1/\sqrt{n})$).
\item The pair potentials $\cL_{uw}(x_u,x_w)$ depend only on the group ratio $x_u x_w^{-1}$. (This is the core property of a \emph{synchronization} problem.)
\item The pair potentials $\cL_{uw}$ are band-limited as a function of $x_u x_w^{-1}$. This assumption (or approximation) allows the algorithm to track only finitely many Fourier coefficients. 
\end{itemize}
The formulation of AMP for general models of this form is discussed in the next section.

\section{AMP over general compact groups}\label{sec:alg}

The approach discussed above for $U(1)$-synchronization with multiple frequencies readily generalizes to the setting of an arbitrary compact group $G$, with Fourier theory generalized to the representation theory of $G$. Just as the Fourier characters are precisely the irreducible representations of $U(1)$, we will represent distributions over $G$ by an expansion in terms of irreducible representations, as described by the Peter--Weyl theorem. Under the assumption (or approximation) of band-limited pairwise observations, it will suffice to store a finite number of coefficients of this expansion. (Note that finite groups have a finite number of irreducible representations and so the band-limited requirement poses no restriction in this case.)

A geometric view on this is as follows. Belief propagation ideally sends messages in the space of distributions on $G$; this is a form of formal convex hull on $G$, and is illustrated in the case of $\ZZ/2$-synchronization by sending messages valued in $[-1,1]$. When $G$ is infinite, however, this space is infinite-dimensional and thus intractable. We could instead ask whether the convex hull of $G$ taken in some finite-dimensional embedding is a sufficient domain for messages. The key to our approach is the observation that, when observations are band-limited, it suffices to take an embedding of $G$ described by a sum of irreducible representations.

This section will be devoted to presenting our AMP algorithm in full generality, along with the synchronization model that it applies to. In particular, the algorithm can run on the general graphical model formulation of Section~\ref{sec:general-model}, but when we analyze its performance we will restrict to the Gaussian observation model of Section~\ref{sec:gaussian}.

\subsection{Representation theory preliminaries}

\subsubsection{Haar measure}

A crucial property of compact groups is the existence of a (normalized) Haar measure, a positive measure $\mu$ on the group that is invariant under left and right translation by any group element, normalized such that $\mu(G) = 1$ \cite{brocker-tomdieck}. This measure amounts to a concept of `uniform distribution' on such a group, and specializes to the ordinary uniform distribution on a finite group. Throughout this paper, integrals of the form
$$ \int_G f(g) \,\dee g, $$
are understood to be taken with respect to Haar measure.

\subsubsection{Peter--Weyl decomposition}

Fix a compact group $G$. We will be working with the density functions of distributions over $G$. In order to succinctly describe these, we use the representation-theoretic analogue of Fourier series: the Peter--Weyl decomposition. The Peter--Weyl theorem asserts that $L^2(G)$ (the space of square-integrable, complex scalar functions on $G$) is the closure of the span (with coefficients from $\CC$) of the following basis, which is furthermore orthonormal with respect to the Hermitian inner product on $L^2(G)$ \cite{brocker-tomdieck}:
$$ R_{\rho ab}(g) = \sqrt{d_\rho}\, \rho(g)_{ab}, $$
indexed over all complex irreducible representations $\rho$ of $G$, and all $1 \leq a \leq d_\rho$, $1 \leq b \leq d_\rho$ where $d_\rho = \dim \rho$. The representations are assumed unitary (without loss of generality). The inner product is taken to be conjugate-linear in the second input.

Since we want our algorithm to be able to store the description of a function using finite space, we fix a finite list $\mathcal{P}$ of irreducible representations to use. From now on, all Peter--Weyl decompositions will be assumed to only use representations from $\mathcal{P}$; we describe functions of this form as \emph{band-limited}. We 
exclude the trivial representation from this list because we will only need to describe functions up to an additive constant. Given a real-valued function $f: G \to \RR$, we will often write its Peter--Weyl expansion in the form
$$ f(g) = \sum_\rho \langle \hat{f}_\rho, R_\rho(g) \rangle, $$
where $\hat{f}_\rho$ and $R_\rho(g) = \sqrt{d_\rho}\, \rho(g)$ are $d_\rho \times d_\rho$ complex matrices. Here $\rho$ ranges over the irreducibles in $\mathcal{P}$; we assume that the functions $f$ we are working with can be expanded in terms of only these representations. The matrix inner product used here is defined by $\langle A,B \rangle = \mathrm{Tr}(AB^*)$. The Peter--Weyl coefficients of a function can be extracted by integration against the appropriate basis functions:
$$\hat{f}_\rho = \int_G R_\rho(g) f(g) \,\dee g.$$

\noindent By analogy to Fourier theory, we will sometimes refer to the coefficients $\hat{f}_\rho$ as \emph{Fourier  coefficients}, and refer to the irreducible representations as \emph{frequencies}.

\subsubsection{Representations of real, complex, quaternionic type}
\label{sec:types}

Every irreducible complex representation of a compact group $G$ over $\CC$ is of one of three types: real type, complex type, or quaternionic type \cite{brocker-tomdieck}. We will need to deal with each of these slightly differently. In particular, for each type we are interested in the properties of the Peter--Weyl coefficients that correspond to the represented function being real-valued.

A complex representation $\rho$ is of \emph{real type} if it can be defined over the reals, i.e.\ it is isomorphic to a real-valued representation. Thus we assume without loss of generality that we are working with a real-valued $\rho$. In this case it is clear that if $f$ is a real-valued function then (by integrating against $R_\rho$) $\hat{f}_\rho$ is real. Conversely, if $\hat{f}_\rho$ is real then the term $\langle \hat{f}_\rho, R_\rho(g)\rangle$ (from the Peter--Weyl expansion) is real.

A representation $\rho$ is of \emph{complex type} if $\rho$ is not isomorphic to its conjugate representation $\bar{\rho}$, which is the irreducible representation defined by $\bar{\rho}(g) = \bar{\rho(g)}$. We will assume that the complex-type representations in our list $\mathcal{P}$ come in pairs, i.e.\ if $\rho$ is on the list then so is $\bar{\rho}$. If $f$ is real-valued, we see (by integrating against $R_\rho$ and $\bar{R_\rho}$) that $\hat{f}_\rho = \bar{\hat{f}_{\bar{\rho}}}$. Conversely, if $\hat{f}_\rho = \bar{\hat{f}_{\bar{\rho}}}$ holds then $\langle \hat{f}_\rho, R_\rho(g)\rangle + \langle \hat{f}_{\bar\rho}, R_{\bar\rho}(g)\rangle$ is real.

Finally, a representation $\rho$ is of \emph{quaternionic type} if it can be defined over the quaternions in the following sense: $d_\rho$ is even and $\rho(g)$ is comprised of $2 \times 2$ blocks, each of which encodes a quaternion by the following relation:
$$a + bi + cj + dk \quad\leftrightarrow\quad \left(\begin{array}{cc} a+bi & c+di \\ -c+di & a-bi \end{array} \right).$$
Note that this relation respects quaternion addition and multiplication. Furthermore, quaternion conjugation (negate $b,c,d$) corresponds to matrix conjugate transpose. If a matrix is comprised of $2 \times 2$ blocks of this form, we will call it \emph{block-quaternion}. Now let $\rho$ be of quaternionic type (and assume without loss of generality that $\rho$ takes the above block-quaternion form), and let $f$ be a real function $G \to \mathbb{R}$. By integrating against $R_\rho$ we see that $\hat{f}_\rho$ must also be block-quaternion. Conversely, if $\hat{f}_\rho$ is block-quaternion then $\langle \hat{f}_\rho, R_\rho(g) \rangle$ is real; to see this, write $\langle \hat{f}_\rho, R_\rho(g) \rangle = \mathrm{Tr}(\hat{f}_\rho R_\rho(g)^*)$, note that $\hat{f}_\rho R_\rho(g)^*$ is block-quaternion, and note that the trace of any quaternion block is real.

\subsection{Graphical model formulation}\label{sec:general-model}
As in Section~\ref{sec:bp-amp}, we take the standpoint of probabilistic inference over a graphical model. Thus we consider the task of estimating $g \in G^n$ from observations that induce a posterior probability factoring into pairwise likelihoods:
\begin{equation} \Pr(g) \propto \prod_{\{u,w\}} \cL_{uv}(g_u,g_w) \label{eq:graphical-model} \end{equation}
where $\{u,w\}$ ranges over all (undirected) edges (without self-loops). We assume that the pair interactions $\cL_{uw}(g_u,g_w)$ are in fact a function of $g_u g_w^{-1} \in G$, depending only on the relative orientation of the group elements. Without loss of generality, we can take $\cL_{uu} = 0$ for all $u$. This factorization property amounts to a graphical model for $g$, with each entry $g_u \in G$ corresponding to a vertex $u$, and each pair interaction $\cL_{uw}$ represented by an edge of the model.

Taking a Peter--Weyl decomposition of $\cL_{uw}$ as a function of $g_u g_w^{-1}$ allows us to write:
$$ \cL_{uw}(g_u,g_w) = \exp \sum_{\rho} \left\langle Y^{\rho}_{uw}, \rho(g_u g_w^{-1}) \right\rangle, $$
where $\rho$ runs over all irreducible representations of $G$.
We require coefficients $Y_{uv}^\rho \in \mathbb{C}^{d_\rho \times d_\rho}$ for which this expansion is real-valued (for all $g_u g_v^{-1}$). We also require the symmetry $\cL_{uv}(g_u,g_v) = \cL_{vu}(g_v,g_u)$, which means $Y_{uv}^\rho = (Y_{vu}^\rho)^*$. Let $Y_\rho$ be the $nd_\rho \times nd_\rho$ matrix with blocks $Y_{uv}^\rho$.

The input to our synchronization problem will simply be the coefficients $Y_\rho$. These define a posterior distribution $\mu$ on the latent vector $g$ of group elements, and our goal is to approximately recover $g$ up to a global right-multiplication by some group element.

We suppose that the observations are \emph{band-limited}: $Y_\rho = 0$ except on a finite set $\cP$ of irreducible representations. This will allow us to reduce all Peter--Weyl decompositions to a finite amount of relevant information. We will always exclude the trivial representation from $\cP$: this representation can only contribute a constant factor to each pair likelihood, which then disappears in the normalization, so that without loss of generality we can assume the coefficient of the trivial representation to always be zero.

One might also formulate a version of this model that allows node potentials:
$$ \Pr(g) \propto \left( \prod_{\{u,w\}} \cL_{uv}(g_u,g_w) \right) \left( \prod_u \cL_u(g_u) \right), $$
expressing a nontrivial prior or observation on each group element. Although this model is compatible with our methods (so long as the node potentials are also band-limited), we suppress this generality for the sake of readability.

Many synchronization problems (for instance, sensor localization) have noise on each pairwise measurement, and fit this graphical model formulation well. Other synchronization problems (for instance, cryo-EM) are based on per-vertex measurements with independent randomness; one can derive pairwise information by comparing these measurements, but these pairwise measurements do not have independent noise and do not strictly fit the model described above. We note that prior work has often run into the same issue and achieved strong results nonetheless. Indeed, this model is closely related to the ``non-unique games'' model of \cite{nug} for which the application to cryo-EM is discussed; this model suggests minimizing an objective of the form (\ref{eq:graphical-model}) without interpreting it as a posterior likelihood. We defer a close examination of message-passing algorithms for synchronization problems with per-vertex noise to future work.

\subsection{AMP algorithm}

We now state our AMP algorithm. The algorithm takes as input the log-likelihood coefficients $Y_\rho \in \mathbb{C}^{nd_\rho \times nd_\rho}$, for each $\rho$ in a finite list $\cP$ of irreducibles (which must not contain the trivial representation; also for each complex-type representation $\rho$ in the list, $\bar\rho$ must also appear in the list). The algorithm's state at time $t$ is comprised of the Fourier coefficients $C_\rho^{(t)} \in \mathbb{C}^{nd_\rho \times d_\rho}$, which are updated as follows.

\begin{algorithm}[AMP for synchronization over compact groups]---
\begin{itemize}
\item For each $\rho \in \cP$ and each vertex $u$, initialize $C_{u,\rho}^{(0)}$ to small random values in $\CC$.
\item Iterate for $1 \leq t \leq T$:
\begin{itemize}
\item For each $\rho \in \cP$, set
$$ C_{u,\rho}^{(t)} = d_\rho^{-1} \sum_{w \neq u} Y_{uw}^{\rho} V_{w,\rho}^{(t-1)} - d_\rho^{-2} |Y_\mathrm{typ}^{\rho}|^2 V_{u,\rho}^{(t-2)} \sum_w \left(d_\rho I - (V_{w,\rho}^{(t-1)})^* V_{w,\rho}^{(t-1)} \right), $$
the Fourier coefficients of the estimated posterior log-densities, with Onsager correction. Here $|Y_\mathrm{typ}^{\rho}|^2$ denotes the average squared-norm of the entries of $Y^\rho$.
\item For each $u$ and each $\rho \in \cP$, set $V_{u,\rho}^{(t)} = \cE_\rho(C_u^{(t)})$, where as in Section~\ref{sec:amp-u1-mult}, the nonlinear transformation
\begin{equation}\mathcal{E}_\rho(C) = \int_G {R_{\rho}(g)} \exp \left(\sum_{\rho'} \left\langle {C_{\rho'}}, {R_{\rho'}(g)} \right\rangle \right) \dee g \Big/ \int_G \exp \left(\sum_{\rho'} \left\langle {C_{\rho'}}, {R_{\rho'}(g)} \right\rangle \right) \dee g \label{eq:cE} \end{equation}
takes the Fourier coefficients for a function $f$ on $G$ and returns those of $\exp \circ f$, re-normalized to have integral $1$. These $V_{u,\rho}^{(t)}$ are the Fourier coefficients of the estimated posterior densities, truncated to the contribution from irreducibles $\cP$, which suffice for the next iteration.
\end{itemize}
\item Return the posteriors represented by $C_{u,\rho}^{(T)}$, or some rounding of these (e.g. the per-vertex MAP estimate).
\end{itemize}
\end{algorithm}
\noindent This algorithm follows the intuition of Section~\ref{sec:intuition}, and derivations can be found in Sections~\ref{sec:bp-deriv} and~\ref{sec:mmse-deriv}.

Note that each iteration runs in time $O(n^2)$, which is linear in the input matrices. This runtime is due to the matrix--vector products; the rest of the iteration takes $O(n)$ time. We expect $O(\log n)$ iterations to suffice, resulting in a nearly-linear-time algorithm with respect to the matrix inputs. Some applications may derive from per-vertex observations that are pairwise compared to produce edge observations, hacked into this framework by an abuse of probability; our algorithm will then takes nearly-quadratic time with respect to the vertex observations. However, some such per-vertex applications produce matrices with a low-rank factorization $Y_\rho = U_\rho U_\rho^\top$, for which the matrix--vector product can be performed in $O(n)$ time.

\subsection{Gaussian observation model}
\label{sec:gaussian}

Our AMP algorithm handles the general graphical model formulation above, but we will be able to analyze its performance in more detail when restricted to the following concrete Gaussian observation model (which we first introduced in \cite{pwbm-contiguity}), generalizing the Gaussian models of Section~\ref{sec:intuition}. First, latent group elements $g_u$ are drawn independently and uniformly from $G$ (from Haar measure). Then for each representation $\rho$ in $\mathcal{P}$, we observe the $nd_\rho \times nd_\rho$ matrix
$$M_\rho = \frac{{\lambda_\rho}}{n} X_\rho X_\rho^* + \frac{1}{\sqrt{n d_\rho}} W_\rho.$$
Here $X_\rho$ is the $nd_\rho \times d_\rho$ matrix formed by vertically stacking the $d_\rho \times d_\rho$ matrices $\rho(g_u)$ for all vertices $u$. $\lambda_\rho$ is a signal-to-noise parameter for the frequency $\rho$. The noise $W_\rho$ is a Gaussian random matrix drawn from the GOE (Gaussian orthogonal ensemble), GUE, or GSE, depending on whether $\rho$ is of real type, complex type, or quaternionic type, respectively. In any case, $W_\rho$ is normalized so that each off-diagonal entry has expected squared-norm 1. To be concrete, in the real case the entries are $\mathcal{N}(0,1)$ and in the complex case, the real and imaginary parts of each entry are $\mathcal{N}(0,1/2)$. For the quaternionic case, each $2 \times 2$ block encodes a quaternion value $a + bi + cj + dk$ in the usual way (see Section~\ref{sec:types}) where $a,b,c,d$ are $\mathcal{N}(0,1/2)$. The noise matrices $W_\rho$ are independent across representations except when we have a conjugate pair of complex-type representations we draw $M_\rho$ randomly as above and define $M_{\bar\rho} = \bar{M_\rho}$ and $\lambda_{\bar\rho} = \lambda_\rho$. Note that the normalization is such that the signal term has spectral norm $\lambda_\rho$ and the noise term has spectral norm $2$.

Special cases of this model have appeared previously: $\mathbb{Z}/2$ with one frequency \cite{dam,sdp-phase} and $U(1)$ with one frequency \cite{sdp-phase}. In fact, \cite{dam} derives AMP for the $\mathbb{Z}/2$ case and proves that it is information-theoretically optimal. We introduced the general model in \cite{pwbm-contiguity}, which presents some statistical lower bounds.

In Appendix~\ref{app:gaussian-loglikelihood}, we show how the Gaussian observation model fits into the graphical model formulation by deriving the corresponding coefficient matrices $Y_\rho$. In particular, we show that $Y_\rho = d_\rho \lambda_\rho M_\rho$, a scalar multiple of the observed Gaussian matrix.

\subsection{Representation theory of some common examples}\label{sec:rep-examples}

In this section we discuss the representation theory of a few central examples, namely $\ZZ/L$, $U(1)$, and $SO(3)$, and connect the general formalism back to the examples of Section~\ref{sec:intuition}.

\paragraph{Representations of $\ZZ/L$ and $U(1)$.} The irreducible representations of these groups are one-dimensional, described by the discrete Fourier transform and the Fourier series, respectively. $U(1)$ has frequencies indexed by $k \in \mathbb{Z}$, given by $\rho(g) = g^k$ where $g \in U(1)$ (i.e.\ a unit-norm complex number). All of these representations are of complex type. We will say ``$U(1)$ with $K$ frequencies'' to refer to the frequencies $1,\ldots,K$ along with their conjugates, the frequencies $-1,\ldots,-K$. Similarly, if we identify $\mathbb{Z}/L$ with the complex $L$th roots of unity, we have frequencies defined the same way as above, except to avoid redundancy we restrict the range of $k$ as follows. If $L$ is odd, we allow $k \in \{1,2,\ldots,(L-1)/2\}$ along with their conjugates (negations). If $L$ is even, we have complex-type representations $k \in \{1,2,\ldots,L/2 - 1\}$ (along with their conjugates), plus an additional real-type representation $k = L/2$. Again, ``$\mathbb{Z}/L$ with $K$ frequencies'' means we take frequencies $1,\ldots,K$ along with their conjugates (when applicable).

For the case of $\mathbb{Z}/2$ we can now see how the $\tanh$ function from the AMP algorithm of Section~\ref{sec:bp-amp} arises as a special case of the nonlinear transformation $\cE$ occurring in AMP. The only nontrivial representation of $\mathbb{Z}/2$ is the `parity' representation in which $-1$ acts as $-1$. In this context, $\cE$ will first input the Fourier series of a log-density with respect to the uniform measure on $\{\pm 1\}$:
$$ \log \dd[\mu_u]{x} + \mathrm{const.} = c x $$
and evaluate this at $\pm 1$ to obtain the values $\pm c$. We then compute the exponential of this at each point to obtain the un-normalized density of $e^{-c}$ at $-1$ and $e^c$ at $1$. Normalizing, the density values are $e^{-c}/(e^{-c} + e^c)$ and $e^c/(e^{-c} + e^c)$, so that the new parity coefficient is
$$ \cE_{\mathrm{parity}}(c) = \frac{e^c - e^{-c}}{e^c + e^{-c}} = \tanh c. $$

\paragraph{Representations of $SO(3)$.} This group has one irreducible representation $\rho_k$ of each odd dimension $d_k = 2k+1$; thus the $k=0$ representation is the trivial representation $\rho_0(g) = 1$, and the $k=1$ representation is the standard representation of $SO(3)$ as rotations of three-dimensional space. All of these representations are of real type, and may be described as the action of rotations on the $2k+1$-dimensional space of homogeneous degree $k$ spherical harmonics. Frequently in the literature (for instance in molecular chemistry), a complex basis for the spherical harmonics is given, and the representation matrices are the complex-valued Wigner D-matrices; however, the representation can be defined over the reals, as is demonstrated by any real orthogonal basis for the spherical harmonics. See e.g.\ Section II.5 of \cite{brocker-tomdieck} for a more detailed account. As in the cases above, we will often refer to synchronization problems over ``$SO(3)$ with $K$ frequencies'', in which the observations are assumed to be band-limited to the first $K$ nontrivial irreducibles with $1 \leq k \leq K$.

\section{Experimental results}
\label{sec:experiments}

We present a brief empirical exploration of the statistical performance of AMP in various settings, and as compared to other algorithms.

\begin{figure}[!ht]\centering
\begin{minipage}{0.8\textwidth}\centering
\includegraphics[width=0.7\linewidth]{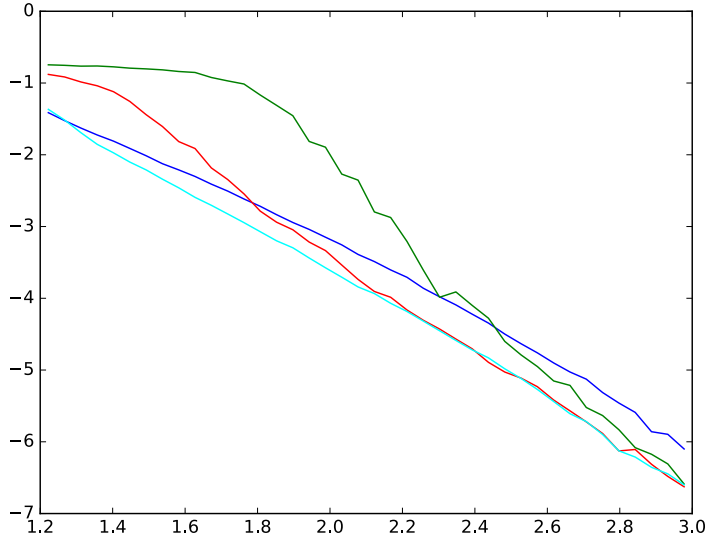}
\caption{Comparison of iterative algorithms for Gaussian $\ZZ/2$ synchronization. The horizontal axis represents the signal-to-noise ratio $\lambda$, and the vertical axis depicts the log-error $\ln(1-|\langle x,\hat x \rangle/n|)$ where $x \in \{\pm 1\}^n$ is the ground truth and $\hat x \in \{\pm 1\}^n$ is the (rounded) output of the algorithm. From top to bottom, as measured on the left side: projected power iteration (green), soft-threshold power iteration (red), spectral method (blue), and AMP (cyan). Each data point is an average of $200$ trials with $n=2000$ vertices.}
\label{fig:z2-compare}
\end{minipage}
\end{figure}

In Figure~\ref{fig:z2-compare} we compare the performance of the spectral method, projected power iteration, soft-threshold power iteration without an Onsager correction, and full AMP (see Sections~\ref{sec:z2} and~\ref{sec:bp-amp}) for Gaussian $\mathbb{Z}/2$ synchronization. The spectral method achieves the optimal threshold of $\lambda = 1$ as to when nontrivial recovery is possible, but does not achieve the optimal correlation afterwards. The projected power method appears to asymptotically achieve the optimal correlation as $\lambda \to \infty$, but performs worse than the spectral method for small $\lambda$. Soft-thresholding offers a reasonable improvement on this, but the full AMP algorithm strictly outperforms all other methods. This reflects the optimality result of \cite{dam} and highlights the necessity for the Onsager term. The gains are fairly modest in this setting, but increase with more complicated synchronization problems.

\begin{figure}[!ht]
    \centering
    \begin{minipage}{0.47\textwidth}
    \includegraphics[width=\linewidth]{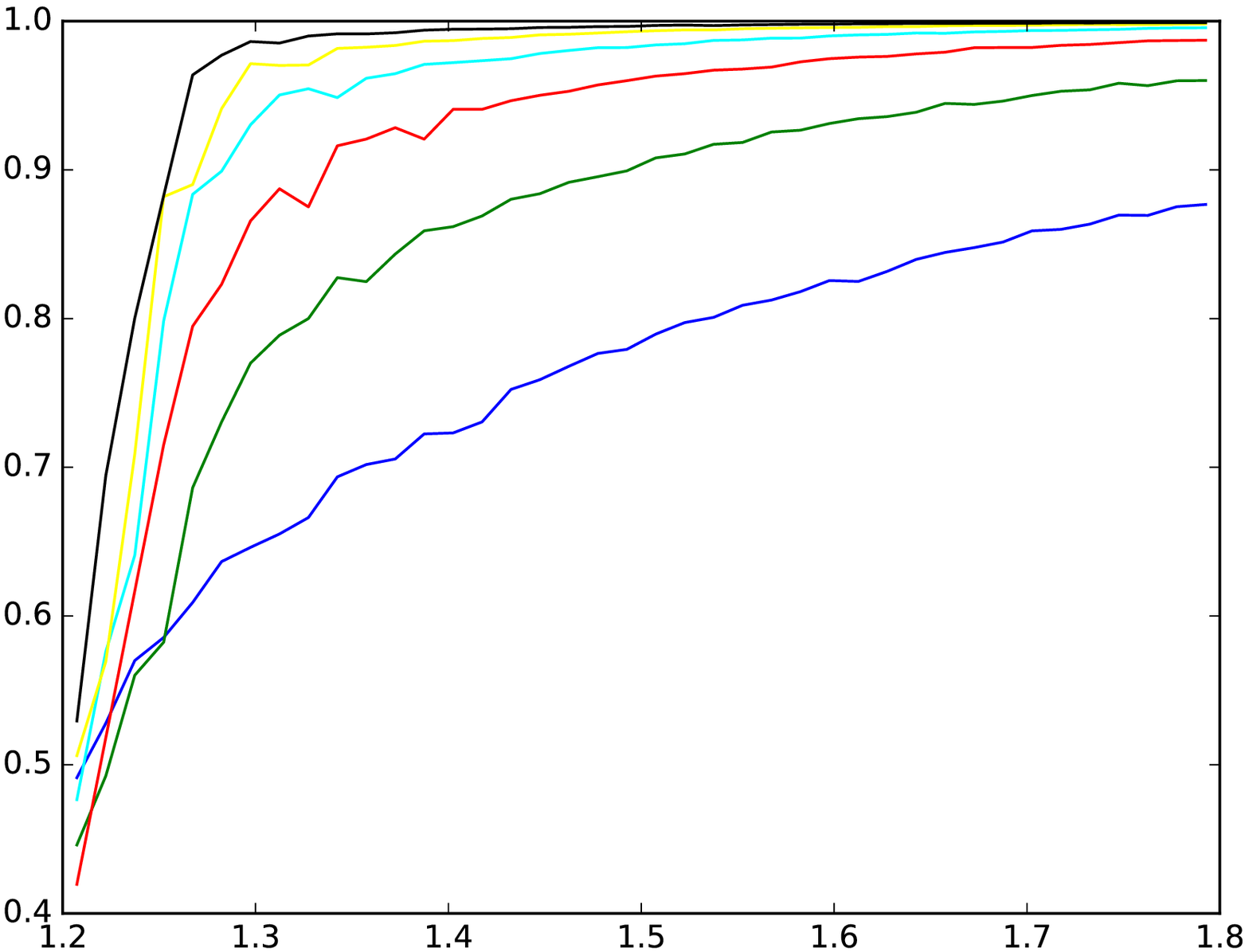}
    \captionof{figure}{Gaussian $U(1)$ synchronization with $K$ frequencies; from bottom to top, $K = 1, \ldots, 6$. The signal-to-noise ratios $\lambda_k$ are all equal, with the common value given by the horizontal axis. Each curve depicts the correlation $|\langle x, \hat x \rangle/n|$ between the ground truth and the AMP estimate. Each data point is an average of $50$ trials with $n=1000$ vertices.}
    \label{fig:u1-many}
    \end{minipage}\hfill%
    \begin{minipage}{0.47\textwidth}
    \includegraphics[width=\linewidth]{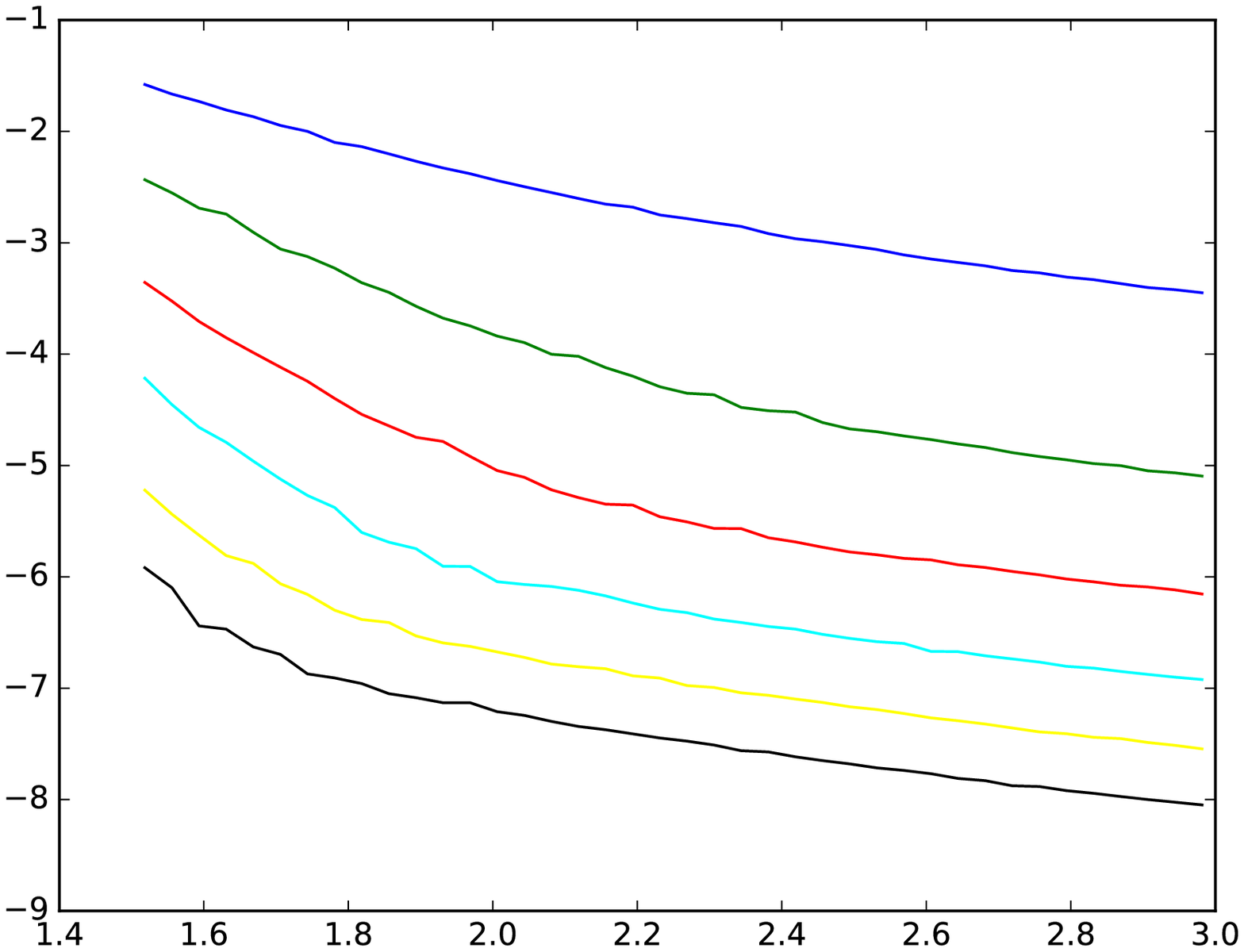}
    \captionof{figure}{Here the vertical axis depicts the log-error $\ln(1-|\langle x, \hat x \rangle/n|)$. From top to bottom: $K = 1, \ldots, 6$.\\[4.5em]}
    \label{fig:u1-many-log}
    \end{minipage}
\end{figure}

Figures~\ref{fig:u1-many} and~\ref{fig:u1-many-log} compare the performance of AMP on Gaussian $U(1)$ synchronization with multiple frequencies; see Section~\ref{sec:amp-u1-mult} for the model. In sharp contrast to spectral methods, which offer no reasonable way to couple the frequencies together, AMP produces an estimate that is orders of magnitude more accurate than what is possible with a single frequency.

In Figures~\ref{fig:so3-many} and~\ref{fig:so3-many-log}, we see similar results over $SO(3)$, under the Gaussian model of Section~\ref{sec:gaussian}. This also demonstrates that AMP is an effective synchronization algorithm for more complicated, non-abelian Lie groups. 

\begin{figure}[!ht]
    \centering
    \begin{minipage}{0.47\textwidth}
    \includegraphics[width=\linewidth]{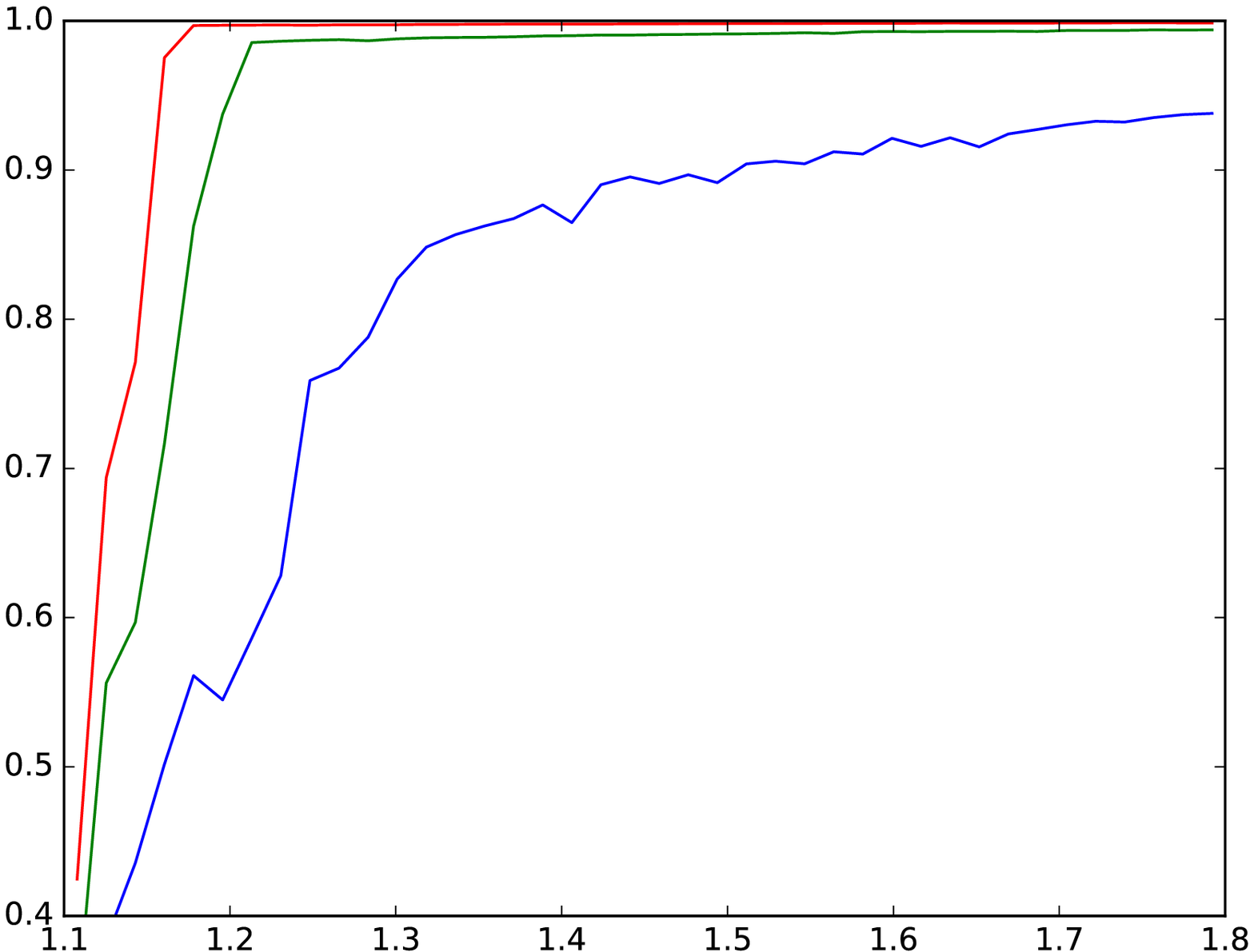}
    \captionof{figure}{Gaussian $SO(3)$ synchronization with $K$ frequencies; from bottom to top, $K = 1, 2, 3$. The signal-to-noise ratios $\lambda_k$ are all equal, with the common value given by the horizontal axis. Each curve depicts the squared correlation $\| X^\top \hat X \|_F/(n\sqrt{3})$ between the ground truth and the AMP estimate. Here $X$ and $\hat X$ are $3n \times n$ matrices where each $3 \times 3$ block encodes an element of $SO(3)$ via the standard representation (rotation matrices). Each data point is an average of $5$ trials with $n=100$ vertices.}
    \label{fig:so3-many}
    \end{minipage}\hfill%
    \begin{minipage}{0.47\textwidth}
    \includegraphics[width=\linewidth]{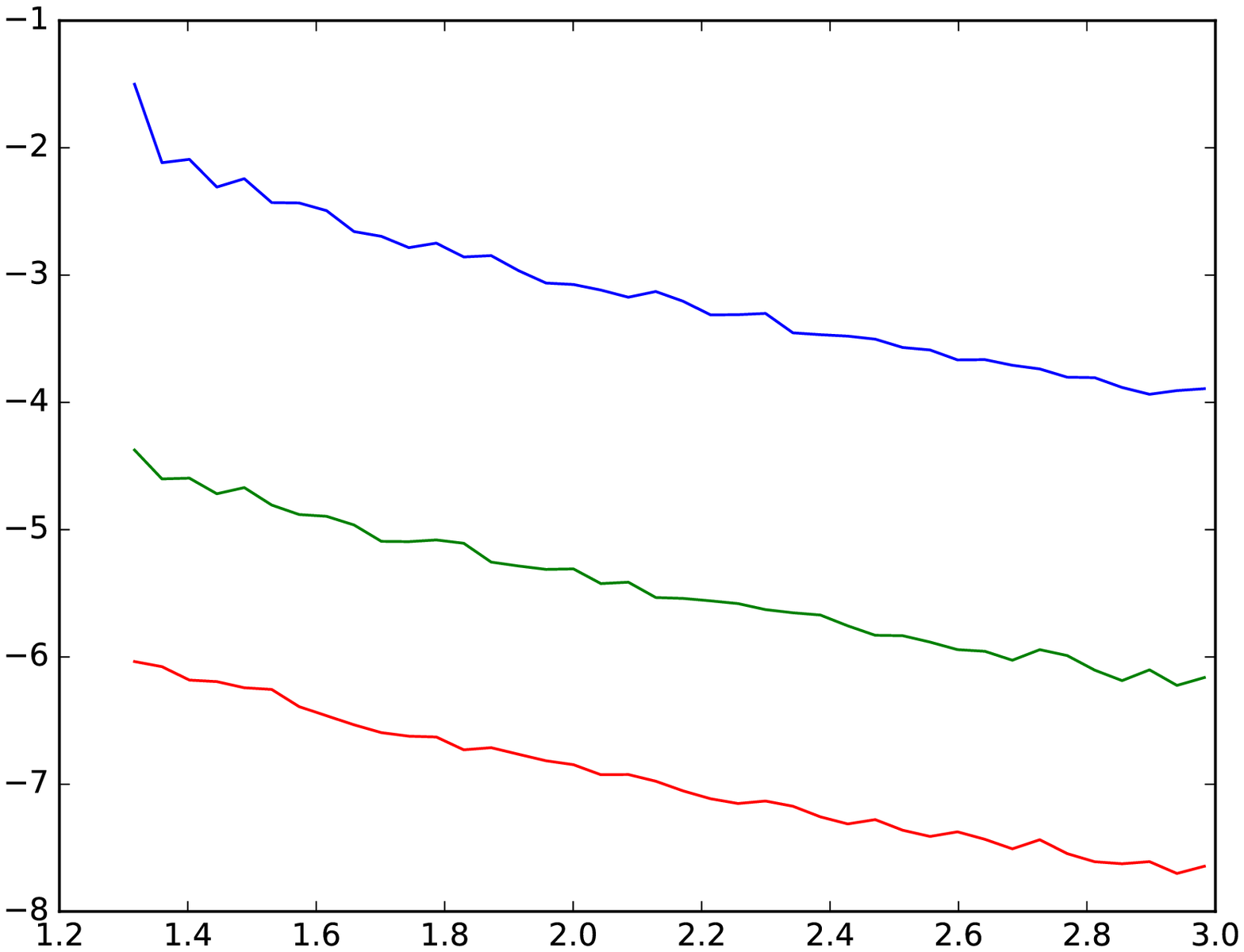}
    \captionof{figure}{Here the vertical axis depicts the log-error $\ln(1-\| X^\top \hat X \|_F/(n\sqrt{3}))$. From top to bottom: $K = 1, 2, 3$.\\[4.5em]}
    \label{fig:so3-many-log}
    \end{minipage}
\end{figure}
This ability to exploit multiple frequencies represents a promising step toward improved algorithms for cryo-electron microscopy, which may be viewed as a synchronization problem over $SO(3)$. Some previous approaches to this problem effectively band-limit the observations to a single frequency and then apply a spectral method \cite{singer-shkolnisky}, and the experiments in Figures~\ref{fig:u1-many}--\ref{fig:so3-many-log} demonstrate that our algorithm stands a compelling chance of achieving a higher-quality reconstruction. 

We remark that some numerical issues arise when computing the nonlinear transformation $\cE$ in our AMP algorithm, which involves integration over the group. Our implementation of $\cE$ for $U(1)$ and $SO(3)$ is based on evaluating each log-density on a discretization of the group, taking the pointwise exponential, and thus approximating each integral by a discrete sum. This approach is somewhat crude but appears to work adequately in our experiments; there is undoubtedly room for this numerical procedure to be improved. More sophisticated methods may be necessary to obtain adequate results on any higher-dimensional Lie groups. Note also that if the vertex posterior in question is extremely concentrated near a point, the numerical value of each integral will depend significantly on whether this spike lies near a discretization point; however, this should affect both the numerator and denominator integrals in (\ref{eq:cE}) by approximately equal factors, so as to have a minimal effect on the normalized value of $\cE_\rho$.


\FloatBarrier

\section{Derivation of AMP from belief propagation}
\label{sec:bp-deriv}


In this section we derive the general AMP algorithm of Section~\ref{sec:alg} starting from belief propagation, similarly to \cite{amp-mot}. We begin with the belief propagation update step (see e.g.\ \cite{mm-book}), writing messages $\mu_{u \to v}^{(t)}$ as densities with respect to Haar measure:
$$ \dd[\mu_{u \to v}^{(t)}]{g_u} = \frac{1}{Z_{u \to v}^{(t)}} \prod_{w \neq u,v} \int_G \cL_{uw}(g_u, g_w) \dd[\mu_{w \to u}^{(t-1)}]{g_w} \,\dee g_w. $$
Here $t$ denotes the timestep and $Z_{u \to v}^{(t)}$ is the appropriate normalization constant. Expand this (positive) probability density as the exponential of an $L^2$ function, expressed as a Peter--Weyl expansion:

$$ \dd[\mu_{u \to v}^{(t)}]{g_u} = \exp \sum_{\rho ab} {C_{u \to v, \rho ab}^{(t)}} \,\bar{R_{\rho ab}(g_u)}. $$

We can extract these Fourier coefficients $C_{u\to v,\rho ab}^{(t)}$ by integrating against the basis functions above. Assume that $\rho$ is not the trivial representation; then:
\begin{align*}
C_{u \to v, \rho ab}^{(t)}
&= \int_G {R_{\rho ab}(g_u)} \log \dd[\mu_{u \to v}^{(t-1)}]{g_u} \,\dee g_u \\
&= \int_G {R_{\rho ab}(g_u)} \log\left( \frac{1}{Z_{u \to v}^{(t-1)}} \prod_{w \neq u,v} \int_G \cL_{uw}(g_u, g_w) \dd[\mu_{w \to u}^{(t-1)}]{g_w} \,\dee g_w \right) \dee g_u \\
&= \int_G {R_{\rho ab}(g_u)} \sum_{w \neq u,v} \log\left( \int_G \exp\left( \sum_{\rho'} \left\langle Y^{\rho'}_{uw}, \rho'(g_u g_w^{-1})\right\rangle \right) \dd[\mu_{w \to u}^{(t-1)}]{g_w} \,\dee g_w \right) \dee g_u.
\end{align*}
As the $Y_{uw}^{\rho'}$ are small, we can pass to a linear expansion about these, incurring $o(1)$ error as $n \to \infty$:
\begin{align*}
&\approx \int_G {R_{\rho ab}(g_u)} \sum_{w \neq u,v} \left( \log \int_G \dd[\mu_{w \to u}^{(t-1)}]{g_w} \,\dee g_w + \sum_{\rho'} \int_G {\left\langle Y^{\rho'}_{uw}, \rho'(g_u g_w^{-1})\right\rangle}  \dd[\mu_{w \to u}^{(t-1)}]{g_w} \,\dee g_w \bigg/ \int_G \dd[\mu_{w \to u}^{(t-1)}]{g_w} \,\dee g_w \right) \dee g_u \\
&= \sum_{w \neq u,v} \sum_{\rho'} \int_G \int_G {R_{\rho ab}(g_u)} {\left\langle Y^{\rho'}_{uw}, \rho'(g_u g_w^{-1})\right\rangle}  \dd[\mu_{w \to u}^{(t-1)}]{g_w} \,\dee g_w \,\dee g_u.
\end{align*}
To progress further, we will expand the middle factor of the integrand:
\begin{align*}
\left\langle Y_{uw}^{\rho'}, \rho'(g_u g_w^{-1}) \right\rangle
&= \left\langle Y_{uw}^{\rho'}, \rho'(g_u) \rho'(g_w)^* \right\rangle \\
&= \sum_{a'b'c'} {Y_{uw,a'b'}^{\rho'}}\; \bar{\rho'(g_u)_{a'c'}}\; {\rho'(g_w)_{b'c'}}.
\end{align*}
Returning to the previous derivation:
\begin{align*}
C_{u \to v, \rho ab}^{(t)}
&= \sum_{w \neq u,v} \sum_{\rho'} \sum_{a'b'c'} \int_G \int_G {R_{\rho ab}(g_u)} Y_{uw,a'b'}^{\rho'} \bar{\rho'(g_u)_{a'c'}} {\rho'(g_w)_{b'c'}} \dd[\mu_{w \to u}^{(t-1)}]{g_w} \,\dee g_w \,\dee g_u \\
&= \sum_{w \neq u,v} \sum_{\rho'} \sum_{a'b'c'} Y_{uw,a'b'}^{\rho'} \int_G {R_{\rho ab}(g_u)} \bar{\rho'(g_u)_{a'c'}}\,\dee g_u \cdot \int_G {\rho'(g_w)_{b'c'}} \dd[\mu_{w \to u}^{(t-1)}]{g_w} \,\dee g_w \\
&= d_{\rho'}^{-1} \sum_{w \neq u,v} \sum_{\rho'} \sum_{a'b'c'} Y_{uw,a'b'}^{\rho'} \int_G {R_{\rho ab}(g_u)} \bar{R_{\rho' a' c'}(g_u)}\,\dee g_u \cdot \int_G {R_{\rho' b' c'}(g_w)} \dd[\mu_{w \to u}^{(t-1)}]{g_w} \,\dee g_w \\
&= d_{\rho'}^{-1} \sum_{w \neq u,v} \sum_{\rho'} \sum_{a'b'c'} Y_{uw,a'b'}^{\rho'} \delta_{\rho,\rho'} \delta_{a,a'} \delta_{b,c'} \int_G {R_{\rho' b' c'}(g_w)} \dd[\mu_{w \to u}^{(t-1)}]{g_w} \,\dee g_w \\
&= d_\rho^{-1} \sum_{w \neq u,v} \sum_{b'} Y_{uw,ab'}^{\rho} \int_G {R_{\rho b' b}(g_w)} \dd[\mu_{w \to u}^{(t-1)}]{g_w} \,\dee g_w.
\end{align*}
In matrix form,
$$ C_{u \to v, \rho}^{(t)} = d_\rho^{-1} \sum_{w \neq u,v} Y_{uw}^{\rho} \int_G {R_{\rho}(g_w)} \dd[\mu_{w \to u}^{(t-1)}]{g_w} \,\dee g_w. $$

Let $\cE : \bigoplus_\rho \CC^{d_\rho \times d_\rho} \to \bigoplus_\rho \CC^{d_\rho \times d_\rho}$ denote the transformation from the nontrivial Fourier coefficients $C_{u \to v,\rho}$ of $\log \dd[\mu_{u \to v}^{(t)}]{g_u}$ to the Fourier coefficients of $\dd[\mu_{u \to v}^{(t)}]{g_u}$. Then we have
$$ C_{u \to v, \rho}^{(t)} = d_\rho^{-1} \sum_{w \neq u,v} Y_{uw}^{\rho} \cE_\rho(C_{w \to u}^{(t-1)}). $$
The map $\cE$ amounts to exponentiation in the evaluation basis, except that the trivial Fourier coefficient is missing from the input, causing an unknown additive shift. This corresponds to an unknown multiplicative shift in the output, which we correct for by noting that $\dd[\mu_{u \to v}^{(t)}]{g_u}$ should normalize to $1$. Thus $\cE$ amounts to exponentiation followed by normalization.

Explicitly, we can let
$$ I_{\rho ab}(C) = \int_G {R_{\rho ab}(g)} \exp \left(\sum_{\rho' a'b'} {C_{\rho' a'b'}} \bar{R_{\rho' a'b'}(g)} \right) \dee g. $$
Then $\cE_{\rho ab}(C) = I_{\rho ab}(C) / I_{\triv}(C)$ where $\triv$ denotes the trivial representation $R_\triv(g) = 1$.

\subsection{Onsager correction}

In this section we complete the derivation of AMP by replacing the non-backtracking nature by an \emph{Onsager correction} term, reducing the number of messages from $n^2$ to $n$. This is similar to the derivation in Appendix~A of \cite{bm}.

In order to remove the non-backtracking nature of the AMP recurrence, let us define

\begin{align*}
C_{u,\rho}^{(t)} &= d_\rho^{-1} \sum_{w \neq u} Y_{uw}^{\rho} \cE_\rho(C_{w \to u}^{(t-1)}) \\
&= C_{u \to v,\rho}^{(t)} + \delta_{u \to v,\rho}^{(t)},
\end{align*}
where $\delta_{u \to v,\rho}^{(t)} = d_\rho^{-1} Y_{uv}^{\rho} \cE_\rho(C_{v \to u}^{(t-1)})$. Then, substituting $C_{w \to u}^{(t-1)} = C_w^{(t-1)} - \delta_{w \to u}^{(t-1)}$, we have
\begin{align*}
C_{u,\rho}^{(t)} &= d_\rho^{-1} \sum_{w \neq u} Y_{uw}^{\rho} \cE_\rho\left( \{ C_{w,\rho'}^{(t-1)} - d_{\rho'}^{-1} Y_{wu}^{\rho'} \cE_{\rho'}(C_{u \to w}^{(t-2)}) \}_{\rho'} \right) \\
&\approx d_\rho^{-1} \sum_{w \neq u} Y_{uw}^{\rho} \cE_\rho(C_{w}^{(t-1)}) - d_\rho^{-1} \sum_{w \neq u} Y_{uw}^{\rho} D\cE_\rho\big|_{C_w^{(t-1)}}\left[ \{ d_{\rho'}^{-1} Y_{wu}^{\rho'} \cE_{\rho'}(C_{u \to w}^{(t-2)}) \}_{\rho'} \right] \\
\intertext{where $D$ denotes the total derivative}
&\approx d_\rho^{-1} \sum_{w \neq u} Y_{uw}^{\rho} \cE_\rho(C_{w}^{(t-1)}) - d_\rho^{-1} \sum_{w \neq u} Y_{uw}^{\rho} D\cE_\rho\big|_{C_{w}^{(t-1)}}\left[ \{ d_{\rho'}^{-1} Y_{wu}^{\rho'} \cE_{\rho'}(C_u^{(t-2)}) \}_{\rho'} \right].
\end{align*}

Under the assumption that $Y$ consists of per-edge $O(n^{-1/2})$ noise and $O(n^{-1})$ signal, the error incurred in these two steps should be $o(1)$. We thus reach an entirely non-backtracking recurrence where the first term is a message-passing step and the second term is the so-called \emph{Onsager correction}. It remains to simplify this. We focus on a single matrix coefficient of the correction:
\begin{align*}
\mathrm{Ons}_{u,\rho ab}^{(t)} &= d_\rho^{-1} \sum_{w \neq u} \sum_c Y_{uw,ac}^{\rho} D\cE_\rho\big|_{C_w^{(t-1)}}\left[ \{ d_{\rho'}^{-1} Y_{wu}^{\rho'} \cE_{\rho'}(C_u^{(t-2)}) \} \right]_{cb} \\
&= d_\rho^{-1} \sum_{w \neq u} \sum_{\rho' cef} Y_{uw,ac}^{\rho} \ddp[\cE_{\rho cb}]{C_{\rho' ef}}\big|_{C_w^{(t-1)}} (d_{\rho'}^{-1} Y_{wu}^{\rho'} \cE_{\rho'}(C_u^{(t-2)}))_{ef} \\
&= d_\rho^{-1} \sum_{w \neq u} \sum_{\rho' cefh} Y_{uw,ac}^{\rho} \ddp[\cE_{\rho cb}]{C_{\rho' ef}}\big|_{C_w^{(t-1)}} d_{\rho'}^{-1} Y_{wu,eh}^{\rho'} \cE_{\rho' hf}(C_u^{(t-2)}).
\intertext{Note that $Y_{wu}^{(\rho)} = (Y_{uw}^{(\rho)})^*$, and that other than this relation, the entries of $Y$ consist of $O(n^{-1/2})$-size independent noise and lower-order signal. Hence, each of the $O(n)$ terms contributing to this sum is $O(n^{-1})$, and is only of this order if $Y_{uw,ac}^{\rho}$ and $Y_{wu,eh}^{\rho'}$ are dependent, i.e.\ if $a=h$, $c=e$, and $\rho = \rho'$. Thus:}
\mathrm{Ons}_{u,\rho ab}^{(t)} &= d_\rho^{-2} \sum_{cf} \cE_{\rho af}(C_u^{(t-2)}) \sum_{w \neq u} |Y_{uw,ac}^{\rho}|^2 \ddp[\cE_{\rho cb}]{C_{\rho' cf}}\big|_{C_w^{(t-1)}}.
\intertext{Similarly to the derivation of \cite{bm}, suppose that $|Y_{uw,ac}^{\rho}|$ depends sufficiently little on $\ddp[\cE_{\rho cb}]{C_{\rho' cf}}\big|_{C_w^{(t-1)}}$ that we can treat $|Y_{uw,ac}^{\rho}|$ as a constant $|Y_\mathrm{typ}^{\rho}|^2$.
Then:}
\mathrm{Ons}_{u,\rho ab}^{(t)} &= d_\rho^{-2} |Y_\mathrm{typ}^{\rho}|^2 \sum_f \cE_{\rho af}(C_u^{(t-2)}) \sum_{w \neq u} \sum_c \ddp[\cE_{\rho cb}]{C_{\rho' cf}}\big|_{C_w^{(t-1)}}.
\end{align*}

An interlude, understanding derivatives of $\cE$: 
$$ \ddp[I_{\rho ab}]{C_{\rho' cd}} = \int_G {R_{\rho a b}}(g) \bar{R_{\rho' c d}(g)} \exp \sum_{\rho'' a'b'} C_{\rho'' a'b'} \bar{R_{\rho'' a'b'}(g)} \,\dee g. $$
In particular,
\begin{align*}
\ddp[I_\triv]{C_{\rho ab}} &= \int_G \bar{R_{\rho ab}(g)} \exp \sum_{\rho' a'b'} C_{\rho' a'b'} \bar{R_{\rho' a'b'}(g)} \,\dee g \\
&= \bar{I_{\rho ab}}.
\end{align*}
Note the following convenient identity:
\begin{align*}
\sum_c \ddp[I_{\rho cb}]{C_{\rho cf}}
&= \int_G \left( \sum_c {R_{\rho c b}(g)} \bar{R_{\rho c f}(g)} \right) \exp \sum_{\rho' a'b'} C_{\rho' a'b'} \bar{R_{\rho' a'b'}(g)} \,\dee g \\
&= d_\rho \int_G \left( \sum_c \rho(g)_{cb} \rho(g^{-1})_{fc} \right) \exp \sum_{\rho' a'b'} C_{\rho' a'b'} \bar{R_{\rho' a'b'}(g)} \,\dee g \\
&= d_\rho \int_G \rho(g^{-1} g)_{fb} \exp \sum_{\rho' a'b'} C_{\rho' a'b'} \bar{R_{\rho' a'b'}(g)} \,\dee g \\
&= d_\rho \delta_{bf} \int_G \exp \sum_{\rho' a'b'} C_{\rho' a'b'} \bar{R_{\rho' a'b'}(g)} \,\dee g \\
&= d_\rho \delta_{bf} I_\triv(C).
\end{align*}
Recalling that $\cE_{\rho ab}(C) = I_{\rho ab}(C) / I_\triv(C)$, we have
\begin{align*}
\sum_c \ddp[\cE_{\rho cb}]{C_{\rho cf}}
&= \frac{ I_\triv \sum_c \ddp[I_{\rho cb}]{C_{\rho cf}} - \sum_c I_{\rho cb} \bar{I_{\rho cf}}}{I_\triv^2} \\
&= d_\rho \delta_{bf} - \sum_c \cE_{\rho cb}(C) \overline{\cE_{\rho cf}(C)} \\
&= \left( d_\rho I - \cE_\rho(C)^* \cE_\rho(C) \right)_{fb}.
\end{align*}

Thus we obtain the following form for the Onsager correction:
$$ \mathrm{Ons}_{u,\rho}^{(t)} = d_\rho^{-2} |Y_\mathrm{typ}^{\rho}|^2 \cE_{\rho}(C_u^{(t-2)}) M_\rho^{(t)}, \quad M_\rho^{(t)} = \sum_w d_\rho I - \cE_\rho(C_w^{(t-1)})^* \cE_\rho(C_w^{(t-1)}), $$
with each AMP iteration reading as
$$ C_{u,\rho}^{(t)} = d_\rho^{-1} \sum_{w \neq u} Y_{uw}^{\rho} \cE_\rho(C_{w}^{(t-1)}) - \mathrm{Ons}_{u,\rho}^{(t)}. $$

\section{MMSE derivation and state evolution}
\label{sec:mmse-deriv}

The goal of this section is to derive the state evolution equations that govern the behavior of AMP on the Gaussian synchronization model of Section~\ref{sec:gaussian} (in the large $n$ limit). Along the way, we will give an alternative derivation of the algorithm (excluding the Onsager term) which shows that the nonlinear function $\cE$ has an interpretation as an MMSE (minimum mean squared error) estimator. This derivation is similar to \cite{dam} and based on ideas first introduced by \cite{amp-cs}. We do not give a proof that the state evolution equations derived here are correct (i.e.\ that AMP obeys them) but we will argue for their correctness in Section~\ref{sec:se-correct}.

\subsection{MMSE estimator}
\label{sec:mmse}


We begin by defining a `scalar' problem: a simplification of the Gaussian synchronization model where we attempt to recover a single group element from noisy measurements. We will be able to analyze the Gaussian synchronization model by connection to this simpler model. (This is the idea of \emph{single letterization} from information theory.) Suppose there is an unknown group element $g$ drawn uniformly from $G$ (Haar measure) and for each irreducible representation $\rho$ in our list $\mathcal{P}$ we are given a measurement $u_\rho = \mu_\rho {\rho(g)} + \sigma_\rho z_\rho$ (for some constants $\mu_\rho, \sigma_\rho$). Here $z_\rho$ is a $d_\rho \times d_\rho$ non-symmetric matrix of Gaussian entries (real, complex, or block-quaternion, depending on the type of $\rho$) with all entries (or blocks) independent and each entry normalized to have expected squared-norm 1. (Note that $z_\rho$ is the same as an off-diagonal block of the matrix $W_\rho$ from Section~\ref{sec:gaussian}.) For $\rho$ of complex type, we only get a measurement $u_\rho$ for one representation in each conjugate pair, and define $u_{\bar\rho} = \bar{u_\rho}$. The MMSE estimator for ${\rho(g)}$ (minimizing the matrix mean squared error $\EE\|\hat{\rho(g)}-\rho(g)\|_F^2$) is simply the conditional expectation


\begin{align*}
\mathbb{E}\left[{\rho(g)}\Big|\{u_q\}_q\right] &= \int_{h \in G} {\rho(h)} \exp\left(-\sum_{q} \frac{1}{2 \sigma_q^2} \|u_q - \mu_q {q(h)} \|_F^2 \right) \Big/ \int_{h \in G} \exp\left(-\sum_{q} \frac{1}{2 \sigma_q^2} \|u_q - \mu_q {q(h)} \|_F^2 \right) \\
&= \int_{h \in G} {\rho(h)} \exp\left(\sum_{q} \frac{\mu_q}{\sigma_q^2} \langle u_q, q(h) \rangle \right) \Big/ \int_{h \in G} \exp\left(\sum_{q} \frac{\mu_q}{\sigma_q^2} \langle u_q, q(h) \rangle \right) \\
&\equiv \mathcal{F}_\rho\left(\left\{\frac{\mu_q}{\sigma_q^2} u_q\right\}_q\right)
\end{align*}
where
$$\mathcal{F}_\rho(\{w_q\}_q) = \int_{h \in G} {\rho(h)} \exp\left(\sum_{q} \langle w_q, q(h) \rangle \right) \Big/ \int_{h \in G} \exp\left(\sum_{q} \langle w_q, q(h) \rangle \right).$$
Here $q$ ranges over irreducible representations in our list $\mathcal{P}$ (which includes both $q$ and $\bar q$ for representations of complex type). The likelihoods used in the above computation are derived similarly to those in Appendix~\ref{app:gaussian-loglikelihood}. We recognize $\mathcal{F}$ as a rescaling of the function $\mathcal{E}$ from the AMP update step.

\subsection{AMP update step}

Consider the Gaussian observation model $M_\rho = \frac{\lambda_\rho}{n} X_\rho X_\rho^* + \frac{1}{\sqrt{nd_\rho}} W_\rho$ from Section~\ref{sec:gaussian}. Similarly to \cite{dam}, the MMSE-AMP update step (without Onsager term) is
$$U_\rho^{t+1} = M_\rho \,\mathcal{F}_\rho\left(\left\{\frac{\mu_q^t}{(\sigma_q^t)^2} U_q^t\right\}_q\right)$$
where $t$ indicates the timestep and $\mu_\rho^t, \sigma_\rho^t$ will be defined based on state evolution below. Here the AMP state $U_\rho^t$ is $n d_\rho \times d_\rho$ with a $d_\rho \times d_\rho$ block for each vertex. $\mathcal{F}_\rho$ is applied to each of these blocks separately. We will motivate this AMP update step below, but notice its similarity to the MMSE estimator above.

\subsection{State evolution}
\label{sec:se}

The idea of state evolution is that the AMP iterates can be approximately modeled as `signal' plus `noise' \cite{amp-cs}. Namely, we postulate that $U_\rho^t = \mu_\rho^t X_\rho + \sigma_\rho^t Z_\rho$ for some constants $\mu_t, \sigma_t$, where $Z_\rho$ is a $nd_\rho \times d_\rho$ Gaussian noise matrix with each $d_\rho \times d_\rho$ block independently distributed like $z_\rho$ (from the scalar model) with $Z_{\bar\rho} = \bar{Z_\rho}$ for conjugate pairs. Recall $X_\rho$ has blocks $\rho(g_u)$, the ground truth. Note that this sheds light on the AMP update step above: at each iteration we are given $U_q^t$, a noisy copy of the ground truth; the first thing we do is to apply the MMSE estimator entrywise.

We will derive a recurrence for how the parameters $\mu_\rho$ and $\sigma_\rho$ change after one iteration. To do this, we assume that the noise $W_\rho$ is independent from $Z_\rho$ at each timestep. This assumption is far from true; however, it turns out that AMP's Onsager term corrects for this (e.g.\ \cite{bm}). In other words, we derive state evolution by omitting the Onsager term and assuming independent noise at each timestep. Then if we run AMP (with the Onsager term and the same noise at each timestep), it behaves according to state evolution. We now derive state evolution:

\begin{align*}
U_\rho^{t+1} &= M_\rho\, \mathcal{F}_\rho\left(\left\{\frac{\mu_q^t}{(\sigma_q^t)^2}U_q^t\right\}_q\right) \\
&= \left(\frac{\lambda_\rho}{n} X_\rho X_\rho^* + \frac{1}{\sqrt{nd_\rho}} W_\rho\right) \mathcal{F}_\rho\left(\left\{\frac{\mu_q^t}{(\sigma_q^t)^2}\left(\mu_q^t X_q + \sigma_q^t Z_q\right)\right\}_q\right) \\
&= \left(\frac{\lambda_\rho}{n} X_\rho X_\rho^* + \frac{1}{\sqrt{nd_\rho}} W_\rho\right) \mathcal{F}_\rho\left(\left\{\gamma_q^t X_q + \sqrt{\gamma_q^t} Z_q\right\}_q\right) \\
\intertext{where $\gamma_q^t = \left(\frac{\mu_q^t}{\sigma_q^t}\right)^2$}
&= \frac{\lambda_\rho}{n} X_\rho X_\rho^*\, \mathcal{F}_\rho\left(\left\{\gamma_q^t X_q + \sqrt{\gamma_q^t} Z_q\right\}_q\right) + \frac{1}{\sqrt{nd_\rho}} W_\rho\, \mathcal{F}_\rho\left(\left\{\gamma_q^t X_q + \sqrt{\gamma_q^t} Z_q\right\}_q\right).
\end{align*}

First focus on the signal term:
$$\frac{\lambda_\rho}{n} X_\rho X_\rho^*\, \mathcal{F}_\rho\left(\left\{\gamma_q^t X_q + \sqrt{\gamma_q^t} Z_q\right\}_q\right) \approx \lambda_\rho X_\rho\, \EE_{g,z_q}\left[\rho(g)^* \,\mathcal{F}_\rho\left(\left\{\gamma_q^t q(g) + \sqrt{\gamma_q^t} z_q\right\}_q\right)\right]$$
where $g$ is drawn from Haar measure on $G$, and $z_q$ is a non-symmetric Gaussian matrix of the appropriate type (as in Section~\ref{sec:mmse}). Define $A_\rho^t \in \mathbb{C}^{d_\rho \times d_\rho}$ to be the second matrix in the expression above:
$$A_\rho^t \equiv \EE_{g,z_q}\left[\rho(g)^* \,\mathcal{F}_\rho\left(\left\{\gamma_q^t q(g) + \sqrt{\gamma_q^t} z_q\right\}_q\right)\right].$$
We will see shortly that $A_\rho^t$ is a multiple $a_\rho^t \in \mathbb{R}$ of the identity and so we can now write the signal term as $\lambda_\rho a_\rho^t X_\rho$. Therefore our new signal parameter is $\mu_\rho^{t+1} = \lambda_\rho a_\rho^t$.

We take a short detour to state some properties of $A_\rho^t$, which we prove in Appendix~\ref{app:A}.
\begin{lemma}
\label{lemma:A}
$A_\rho^t$ is a real multiple of the identity: $A_\rho^t = a_\rho^t I_{d_\rho}$ for some $a_\rho^t \in \mathbb{R}$. Furthermore, we have the following equivalent formulas for $a_\rho^t$:
\begin{enumerate}[(i)]
\item $\EE_{g,z_q}\left[\rho(g)^* \,\mathcal{F}_\rho\left(\left\{\gamma_q^t q(g) + \sqrt{\gamma_q^t} z_q\right\}_q\right)\right]$
\item $\EE_{g,z_q}\left[\mathcal{F}_\rho\left(\cdots\right)^* \,\mathcal{F}_\rho\left(\cdots\right)\right]$
\item $\EE_{z_q}\left[\mathcal{F}_\rho\left(\left\{\gamma_q^t I_{d_q} + \sqrt{\gamma_q^t} z_q\right\}_q\right)\right]$
\item $\EE_{z_q}\left[\mathcal{F}_\rho\left(\cdots\right)^* \,\mathcal{F}_\rho\left(\cdots\right)\right]$
\end{enumerate}
where $\cdots$ denotes the argument to $\mathcal{F}_\rho$ from the previous line.
\end{lemma}

Returning to state evolution, we now focus on the noise term:
$$\frac{1}{\sqrt{nd_\rho}} W_\rho\, \mathcal{F}_\rho\left(\left\{\gamma_q^t X_q + \sqrt{\gamma_q^t} Z_q\right\}_q\right).$$
Each entry of this $n d_\rho \times d_\rho$ matrix is Gaussian. The variance (expected squared-norm) of entry $(i,j)$ is (approximately)
\begin{align*}
\frac{1}{nd_\rho} \sum_{k=1}^{n d_\rho} \left|\mathcal{F}_\rho\left(\left\{\gamma_q^t X_q + \sqrt{\gamma_q^t} Z_q\right\}_q\right)_{k,j}\right|^2
&\approx \frac{1}{d_\rho} \mathbb{E}_{g,z_q} \sum_{k=1}^{d_\rho} \left|\mathcal{F}_\rho\left(\left\{\gamma_q^t q(g) + \sqrt{\gamma_q^t} z_q\right\}_q\right)_{k,j}\right|^2 \\
&= \frac{1}{d_\rho} \EE\left[\mathcal{F}_\rho(\cdots)^* \mathcal{F}_\rho(\cdots)\right]_{jj} \\
&= \frac{1}{d_\rho} (A_\rho^t)_{jj} \\
&= \frac{1}{d_\rho} a_\rho^t.
\end{align*}
We therefore have the new noise parameter $(\sigma_\rho^{t+1})^2 = \frac{a_\rho^t}{d_\rho}$.

To summarize, we now have the state evolution recurrence $\mu_\rho^{t+1} = \lambda_\rho a_\rho^t$ and $(\sigma_\rho^{t+1})^2 = \frac{a_\rho^t}{d_\rho}$.

\subsection{Simplified AMP update step}

Note that the state evolution recurrence implies the relation
$$\frac{\mu_\rho^{t+1}}{(\sigma_\rho^{t+1})^2} = d_\rho \lambda_\rho.$$
Provided our initial values of $\mu_\rho, \sigma_\rho$ satisfy this relation (which can always be arranged by scaling the initial $U_\rho$ appropriately), our AMP update step (without Onsager term) becomes
$$U_\rho^{t+1} = M_\rho \,\mathcal{F}_\rho\left(\left\{d_\rho \lambda_\rho U_q^t\right\}_q\right).$$
This is convenient because we can implement AMP without keeping track of the state evolution parameters $\mu_\rho^t, \sigma_\rho^t$. Also note that this variant of AMP matches the original derivation after the rescaling $C_\rho^t = \sqrt{d_\rho} \lambda_\rho U_\rho^t$ (and excluding the Onsager term).

\subsection{Reduction to single parameter (per frequency)}

We will rewrite the state evolution recurrence in terms of a single parameter per frequency. This parameter will be $\gamma_\rho^t$, which was introduced earlier: $\gamma_\rho^t = \left(\frac{\mu_\rho^t}{\sigma_\rho^t}\right)^2$. Recall the state evolution recurrence $\mu_\rho^{t+1} = \lambda_\rho a_\rho^t$ and $(\sigma_\rho^{t+1})^2 = \frac{a_\rho^t}{d_\rho}$. We therefore have the update step
$$\gamma_\rho^{t+1} = \left(\frac{\mu_\rho^{t+1}}{\sigma_\rho^{t+1}}\right)^2 = \frac{(\lambda_\rho a_\rho^t)^2}{a_\rho^t / d_\rho} = d_\rho \lambda_\rho^2 a_\rho^t.$$
Using part (iii) of Lemma~\ref{lemma:A} we can write this as:
\begin{equation}
\label{eq:se}
\gamma_\rho^{t+1} = \lambda_\rho^2 \,\EE_{z_q}\mathrm{Tr}\,\mathcal{F}_\rho\left(\left\{\gamma_q^t I_{d_q} + \sqrt{\gamma_q^t} z_q\right\}_q\right).
\end{equation}


\noindent This is the final form of our state evolution recurrence. The relation between $\mu_\rho, \sigma_\rho, \gamma_\rho$ can be summarized as $\gamma_\rho = d_\rho \lambda_\rho \mu_\rho = d_\rho^2 \lambda_\rho^2 \sigma_\rho^2$.

We expect that the state evolution recurrence (\ref{eq:se}) exactly governs the behavior of AMP in the large $n$ limit. Although the derivation above was heuristic, we discuss its correctness in Section~\ref{sec:se-correct}. There is a caveat regarding how it should be initialized (see Section~\ref{sec:se-correct}) but in practice we can imagine the initial $\gamma$ value is a small random vector. (Note that the initialization $\gamma = \vec 0$ is problematic because state evolution will never leave zero.) We expect that state evolution converges to some fixed point of the recurrence. Some complications arise if there are multiple fixed points (see Section~\ref{sec:gaps}) but we expect there to be a unique fixed point that is reached from any small initialization. This fixed point $\gamma^*$ describes the output of AMP in the sense that (following the postulate of state evolution) the final AMP iterate is approximately distributed as $U_\rho \approx \mu_\rho^* X_\rho + \sigma_\rho^* Z_\rho$, which in terms of $\gamma^*$ is (up to scaling) $U_\rho \approx \gamma_\rho^* X_\rho + \sqrt{\gamma^*_\rho} Z_\rho$. (See \cite{bm} for the precise sense in which we expect this to be true.) Note that one can use this to translate a $\gamma^*$ value into any measure of performance, such as MSE. This gives an exact asymptotic characterization of the performance of AMP for any set of $\lambda_\rho$ values. The most prominent feature of AMP's performance is the threshold at $\lambda = 1$, which we derive in the next section.

One can check that our state evolution recurrence matches the Bayes-optimal cavity and replica predictions of \cite{sdp-phase} for $\mathbb{Z}/2$ and $U(1)$ with one frequency. Indeed, we expect AMP to be statistically optimal in these settings (and many others too; see Section~\ref{sec:gaps}), and this has been proven rigorously for $\mathbb{Z}/2$ \cite{dam}.

\subsection{\texorpdfstring{Threshold at $\lambda = 1$}{Threshold at lambda=1}}
\label{sec:lambda1}

In this section we use the state evolution occurrence to derive the threshold above which AMP achieves nontrivial recovery. In particular, if $\lambda_\rho < 1$ for all frequencies $\rho$ then the AMP fixed point $\gamma^*$ is equal to the zero vector and so AMP gives trivial performance (random guessing) in the large $n$ limit. On the other hand, if $\lambda_\rho > 1$ for at least one frequency $\rho$ then $\gamma^*$ is nonzero and AMP achieves nontrivial recovery.

The zero vector is always a fixed point of state evolution. Whether or not AMP achieves nontrivial performance depends on whether the zero vector is a stable or unstable fixed point. Therefore we consider the regime where $\gamma_\rho$ is small for all $\rho$. When the input $\{w_q\}_q$ to $\mathcal{F}_\rho$ is small, we can approximate $\mathcal{F}_\rho$ by its linearization.
$$\mathcal{F}_\rho\left(\left\{w_q\right\}_q\right) \approx \int_h \rho(h)\left[1 + \sum_q \langle w_q, q(h) \rangle \right] = \int_h \rho(h)\sum_q \langle w_q, q(h) \rangle$$
and so
\begin{align*}
\mathcal{F}_\rho\left(\left\{w_q\right\}_q\right)_{ab} &\approx \int_h \rho(h)_{ab} \sum_{qcd} w_{qcd} \bar{q(h)_{cd}} \\
&= \sum_{qcd} w_{qcd} \int_h \rho(h)_{ab} \bar{q(h)_{cd}} \\
&= \sum_{qcd} w_{qcd} \frac{1}{d_\rho} \delta_{\rho a b, q c d} \\
&= \frac{w_{\rho a b}}{d_\rho}
\end{align*}
which means $\mathcal{F}_\rho\left(\left\{w_q\right\}_q\right) \approx \frac{w_\rho}{d_\rho}$. Now the state evolution update step becomes
\begin{align*}
\gamma_\rho^{t+1} &= \lambda_\rho^2 \,\EE_{z_q}\mathrm{Tr}\,\mathcal{F}_\rho\left(\left\{\gamma_q^t I_{d_q} + \sqrt{\gamma_q^t} z_q\right\}_q\right) \\
&\approx \lambda_\rho^2 \,\EE_{z_q}\mathrm{Tr}\,\frac{1}{d_\rho}\left(\gamma_\rho^t I_{d_\rho} + \sqrt{\gamma_\rho^t} z_q\right) \\
&= \lambda_\rho^2 \gamma_\rho^t.
\end{align*}

\noindent This means that when $\gamma$ is small (but nonzero), $\gamma_\rho$ shrinks towards zero if $\lambda_\rho < 1$ and grows in magnitude if $\lambda_\rho > 1$. We conclude the threshold at $\lambda = 1$.

\section{Correctness of state evolution?}
\label{sec:se-correct}

In this section we justify the heuristic derivation of state evolution in the previous section and argue for its correctness. We first discuss prior work that provides a rigorous foundation for the methods we used, in related settings. We then show numerically that our AMP algorithm obeys the state evolution equations.

\subsection{Rigorous work on state evolution}

State evolution was introduced along with AMP by \cite{amp-cs}. It was later proven rigorously that AMP obeys state evolution in the large $n$ limit (in a particular formal sense) for certain forms of the AMP iteration \cite{bm,jm}. In particular, $\mathbb{Z}/2$ synchronization with Gaussian noise (a special case of our model) falls into this framework and thus admits a rigorous analysis \cite{dam}. Although the proofs of \cite{bm,jm} only consider the case of real-valued AMP, it has been stated \cite{camp} that the proof extends to the complex-valued case. This covers our synchronization model over $U(1)$ with one frequency. In order to cover our general formulation of AMP over any group with any number of frequencies, one needs to replace the complex numbers by a different real algebra (namely a product of matrix algebras). We expect that this generalization should follow from the existing methods.

There is, however, an additional caveat involving the initialization of state evolution. In practice, we initialize AMP to small random values. Recall that we only need to recover the group elements up to a global right-multiplication and so there exists a favorable global right-multiplication so that our random initialization has some correlation with the truth. However, this correlation is $o(1)$ and corresponds to $\gamma = \vec 0$ in the large $n$ limit. This means that technically, the formal proof of state evolution (say for $\mathbb{Z}/2$) tells us that for any fixed $t$, AMP achieves $\gamma = \vec 0$ after $t$ iterations in the large $n$ limit. Instead we would like to show that after $\omega(1)$ iterations we achieve a nonzero $\gamma$. It appears that proving this would require a non-asymptotic analysis of AMP, such as \cite{amp-finite}. It may appear that this initialization issue can be fixed by initializing AMP with a spectral method, which achieves $\Omega(1)$ correlation with the truth; however this does not appear to easily work due to subtle issue about correlation between the noise and iterates. In practice, the initialization issue is actually not an issue at all: with a small random initialization, AMP consistently escapes from the trivial fixed point (provided some $\lambda$ exceeds 1). One way to explain this is that when the AMP messages are small, the nonlinear function $\mathcal{F}$ is essentially the identity (see Section~\ref{sec:lambda1}) and so AMP is essentially just the power method; this roughly means that AMP automatically initializes itself to the output of the spectral method.

\subsection{Experiments on state evolution}

We now present experimental evidence that AMP obeys the state evolution equations. In Figure~\ref{fig:se-correct} we show two experiments, one with $U(1)$ and one with $SO(3)$. In both cases we see that the performance of AMP closely matches the state evolution prediction. We see some discrepancy near the $\lambda = 1$ threshold, which can be attributed to the fact that here we are running AMP with finite $n$ whereas state evolution describes the $n \to \infty$ behavior.

\begin{figure}[!ht]
    \centering
    \begin{subfigure}[t]{0.47\textwidth}
        \centering
        \includegraphics[width=\linewidth]{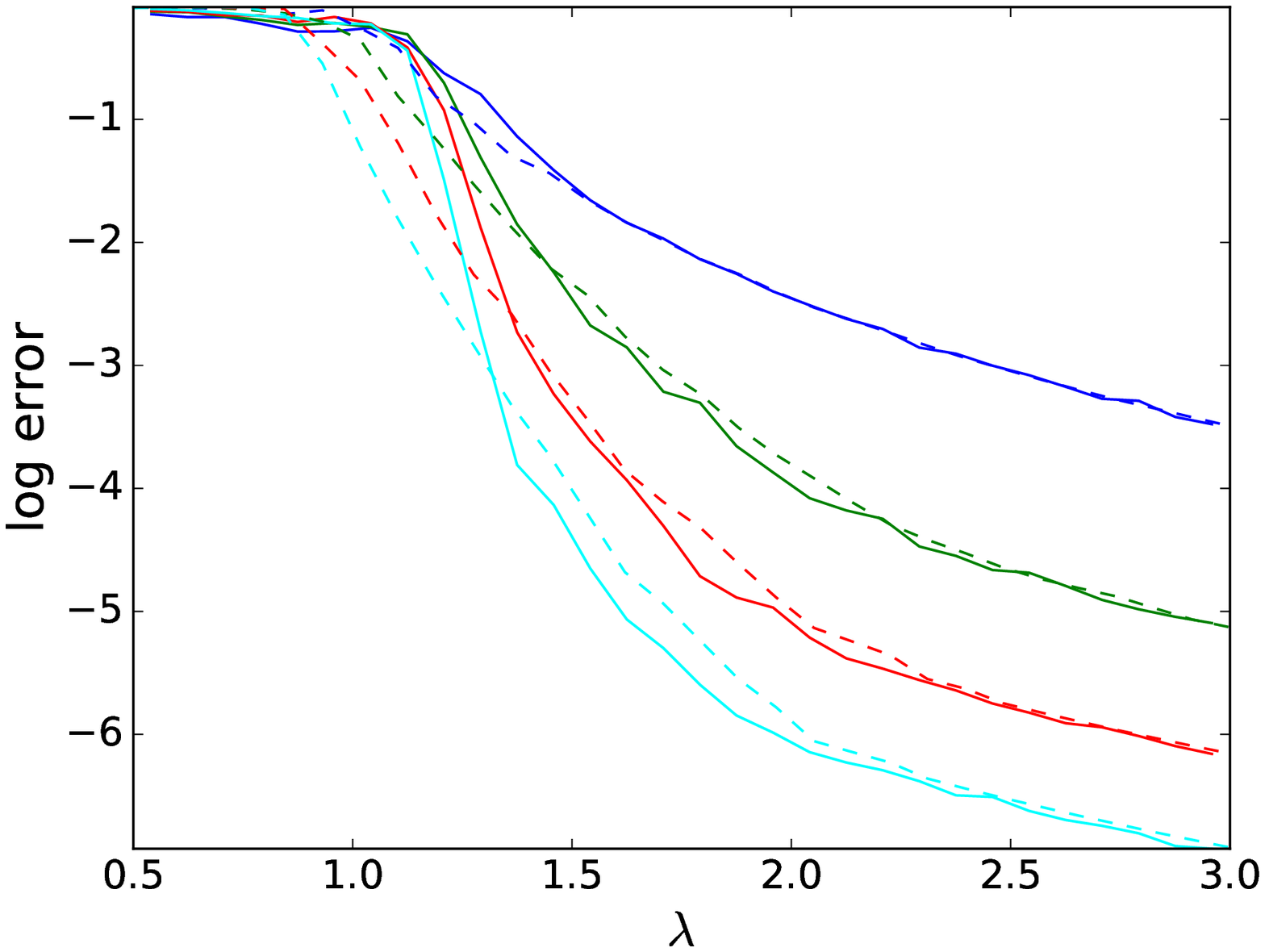}
    \end{subfigure}
    \hfill
    \begin{subfigure}[t]{0.47\textwidth}
        \centering
        \includegraphics[width=\linewidth]{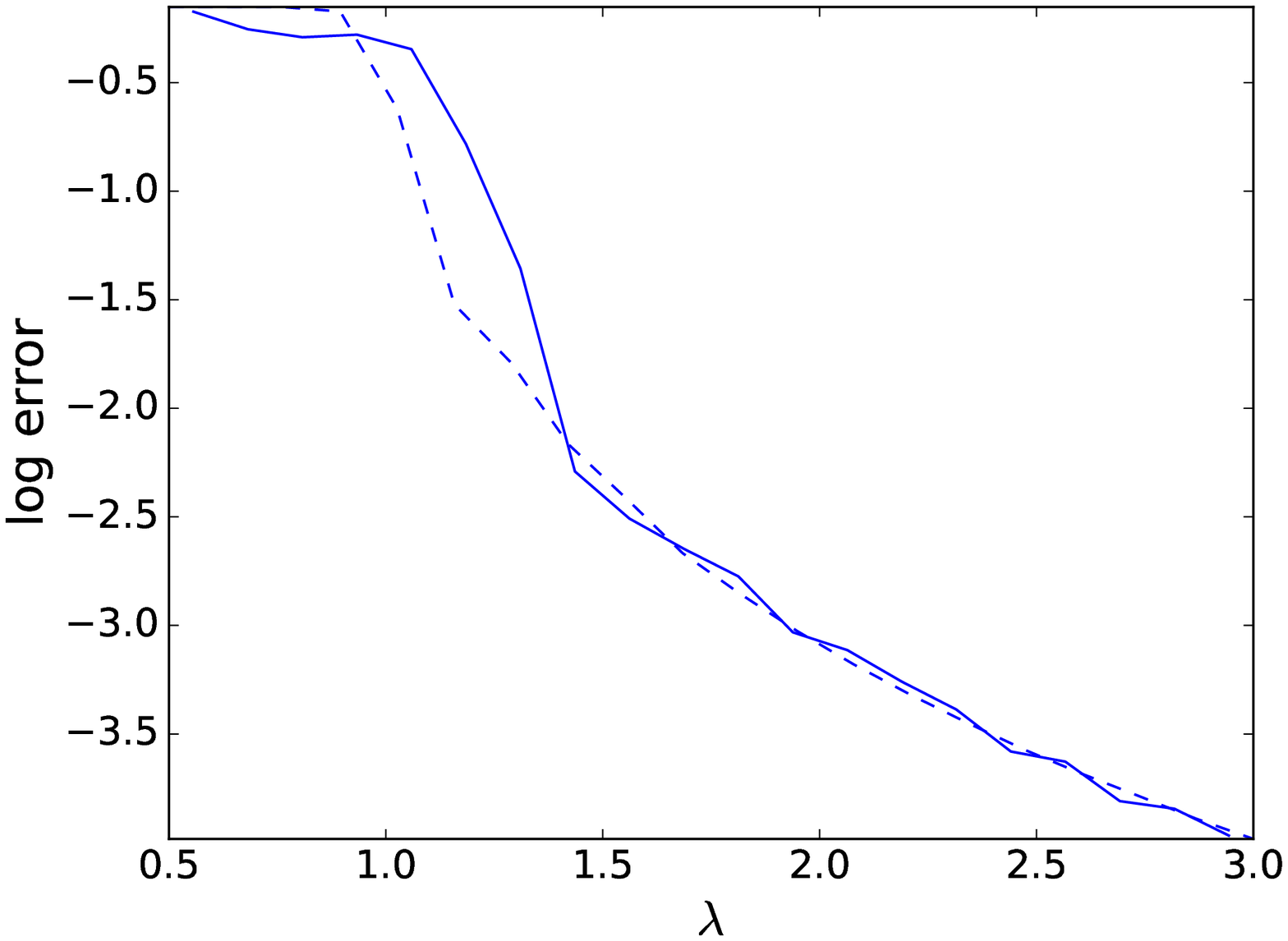}
    \end{subfigure}
    \caption{AMP compared to the state evolution equations experimentally. {\bf Left}: $U(1)$ with $K$ frequencies, for $K = 1,2,3,4$ (from top to bottom) with $n=100$. The solid line is AMP and the dotted line is the state evolution prediction. The horizontal axis is the signal-to-noise ratio $\lambda$, which we take to be equal on all frequencies. The vertical axis is the natural logarithm of error, which is defined as $\mathrm{error} = 1-|\langle x, \hat x \rangle/n| \in [0,1]$ where $x \in U(1)^n$ is the truth and $\hat x \in U(1)^n$ is the (rounded) output of AMP. In particular, a log error value of zero (top of the figure) indicates trivial recovery (random guessing), and lower values are better. {\bf Right}: $SO(3)$ with one frequency, with $n = 50$. Now error is measured as $\mathrm{error} = 1-\frac{1}{\sqrt{ 3}n}\| X^\top \hat X \|_F \in [0,1]$ where $X,\hat X$ are $3n \times n$ matrices whose $3 \times 3$ blocks encode elements of $SO(3)$ via the standard representation (3D rotation matrices).}
    \label{fig:se-correct}
\end{figure}

\section{Statistical-to-computational gaps}
\label{sec:gaps}

In various settings it has been shown, using standard but non-rigorous methods from statistical physics, that the analysis of AMP and state evolution yields a complete picture of the various ``phase transitions'' that occur in a computational problem (e.g.\ \cite{phase-sparse,mmse-low-rank}). In some settings, certain features of these predictions have been confirmed rigorously (e.g.\ \cite{mi-rank-one,mi-rank-one-proof}). In this section we will use these methods to give non-rigorous predictions about statistical-to-computational gaps in the Gaussian synchronization model.

In Section~\ref{sec:lambda1} we have seen that (in the large $n$ limit) AMP achieves nontrivial recovery if and only if $\lambda > 1$ on at least one frequency. In this section, we will see that it is sometimes statistically possible to succeed below this threshold, although no known efficient algorithm achieves this. A rigorous analysis of an inefficient estimator has indeed confirmed that the $\lambda = 1$ threshold can be beaten in some cases \cite{pwbm-contiguity}; the non-rigorous computations in this section give sharp predictions for exactly when this is possible.

\subsection{Free energy}

Recall the parameter $\gamma = \{\gamma_\rho\}_\rho$ from the state evolution recurrence (\ref{eq:se}); $\gamma$ captures the amount of information that AMP's current state has about each frequency, with $\gamma_\rho = 0$ indicating no information and $\gamma_\rho \to \infty$ indicating complete knowledge.

An important quantity is the \emph{Bethe free energy} per variable (also called the \emph{replica symmetric potential function}) of a state $\gamma$, which for the Gaussian synchronization model is given (up to constants) by
$$f(\gamma) = -\frac{1}{4} \sum_\rho d_\rho^2 \lambda_\rho^2 + \frac{1}{2} \sum_\rho d_\rho \gamma_\rho + \frac{1}{4} \sum_\rho \frac{\gamma_\rho^2}{\lambda_\rho^2} - \EE_z \log \EE_g \exp\left(\sum_\rho \langle \rho(g), \gamma_\rho I_{d_\rho} + \sqrt{\gamma_\rho} z_\rho \rangle\right)$$
where $z_\rho$ is a $d_\rho \times d_\rho$ matrix of i.i.d.\ standard Gaussians (of the appropriate type: real, complex, or quaternionic, depending on $\rho$), and $g$ is drawn from Haar measure on the group. We do not include the derivation of this expression, but it can be computed from belief propagation (as in \cite{mmse-low-rank}) or from the replica calculation (as in \cite{sdp-phase}).

Roughly speaking, the interpretation of the Bethe free energy is that it is the objective value that AMP is trying to minimize. AMP can be thought of as starting from the origin $\gamma = 0$ and performing na\"ive gradient descent in the free energy landscape until it reaches a local minimum; the value of $\gamma$ at this minimum describes the final state of AMP. (It can be shown that the fixed points of the state evolution recurrence (\ref{eq:se}) are precisely the stationary points of the Bethe free energy.) As is standard for these types of problems, we conjecture that AMP is optimal among all polynomial-time algorithms. However, with no restriction on efficiency, the information-theoretically optimal estimator is given by the global minimum of the free energy. (This has been shown rigorously for the related problem of rank-one matrix estimation \cite{mi-rank-one-proof}.) The intuition here is that the optimal estimator should use exhaustive search to enumerate all fixed points of AMP and return the one of lowest Bethe free energy. Note that just because we can compute the $\gamma$ value that minimizes the Bethe free energy it does not mean we can achieve this $\gamma$ with an efficient algorithm; $\gamma$ represents correlation between the AMP iterates and the ground truth, and since the truth is unknown it is hard to find iterates that have a prescribed $\gamma$.

\subsection{Examples}

We now examine the Bethe free energy landscapes of some specific synchronization problems at various values of $\lambda$, and discuss the implications. Our primary examples will be $U(1)$ and $\mathbb{Z}/L$ with various numbers of frequencies, as discussed in Section~\ref{sec:rep-examples}. Recall that references to $U(1)$ or $\ZZ/L$ ``with $K$ frequencies'' means that observations are band-limited to the Fourier modes $e^{i k \theta}$ with $|k| \leq K$.

Our first example is $U(1)$ with a single frequency, shown in Figure~\ref{fig:f-u1}. Here we see that the problem transitions from (statistically) `impossible' to `easy' (AMP achieves nontrivial recovery) at $\lambda = 1$, with no (computationally) `hard' regime. In particular, AMP is statistically optimal for every value of $\lambda$.

\begin{figure}[!ht]
    \centering
    \begin{subfigure}[t]{0.47\textwidth}
        \centering
        \includegraphics[width=\linewidth]{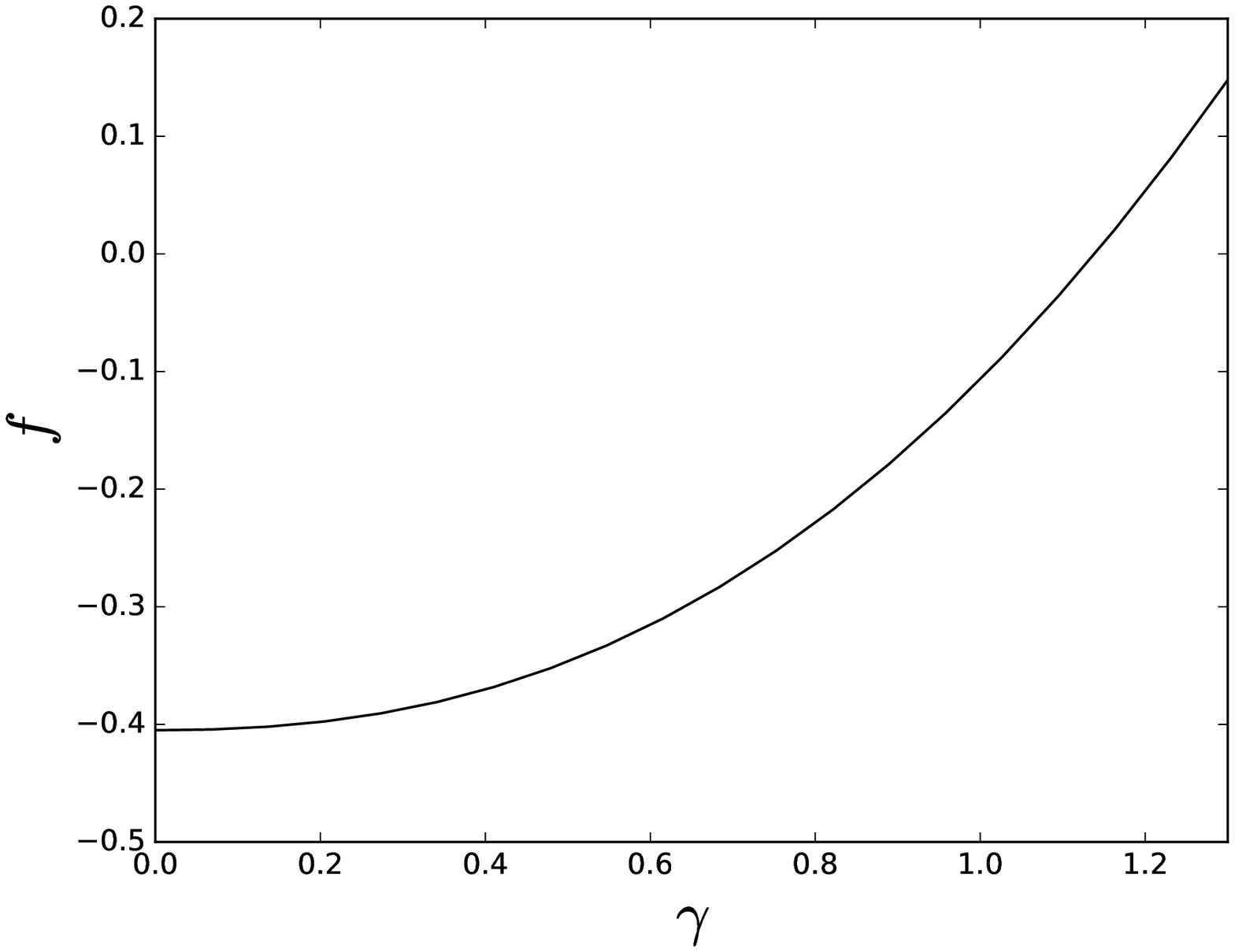}
    \end{subfigure}
    \hfill
    \begin{subfigure}[t]{0.47\textwidth}
        \centering
        \includegraphics[width=\linewidth]{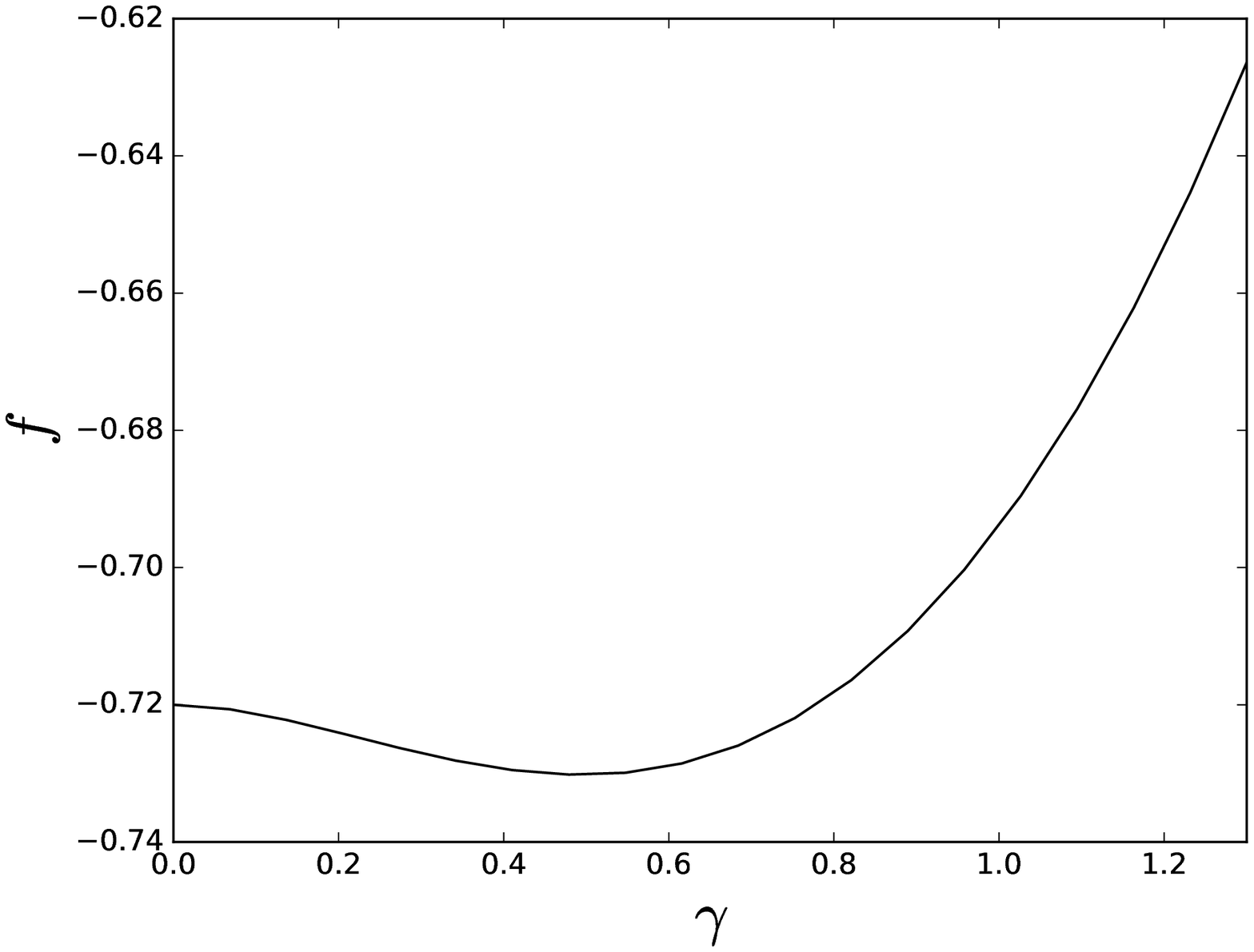}
    \end{subfigure}
    \caption{Free energy landscape for $U(1)$ with 1 frequency. {\bf Left}: $\lambda < 1$. The global minimum of free energy occurs at $\gamma = 0$, indicating that AMP or any other estimator achieves zero correlation with the truth. {\bf Right}: $\lambda > 1$. Now the global minimum occurs at nonzero $\gamma$, and this is achieves by AMP. Therefore AMP achieves the statistically optimal MSE (mean squared error). This MSE departs continuously from zero at the $\lambda = 1$ threshold.}
    \label{fig:f-u1}
\end{figure}

Our next example is a single-frequency problem that exhibits a computational gap (a `hard' phase). In Figure~\ref{fig:a4} we take the alternating group $A_4$ with its irreducible 3-dimensional representation as the rotational symmetries of a tetrahedron. When $\lambda > 1$, AMP achieves statistically optimal performance but when $\lambda$ is below 1 but sufficiently large, AMP gives trivial performance while the statistically optimal estimator gives nontrivial performance. This means we have a computational gap, i.e. there are values of $\lambda$ below the AMP threshold ($\lambda = 1$) where nontrivial recovery is statistically possible.

\begin{figure}[!ht]
    \centering
    \begin{subfigure}[t]{0.47\textwidth}
        \centering
        \includegraphics[width=\linewidth]{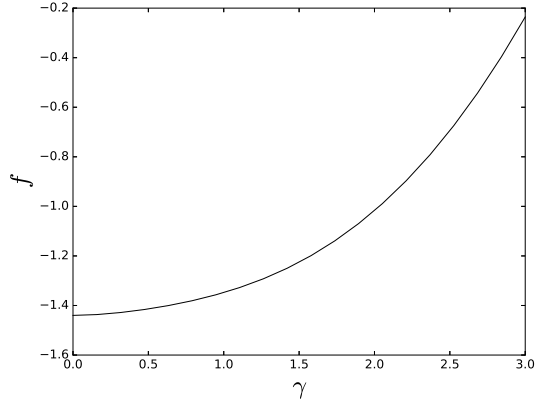}
        \captionof*{figure}{{\bf (a)} $\lambda = 0.8$, impossible}
    \end{subfigure}
    \hfill
    \begin{subfigure}[t]{0.47\textwidth}
        \centering
        \includegraphics[width=\linewidth]{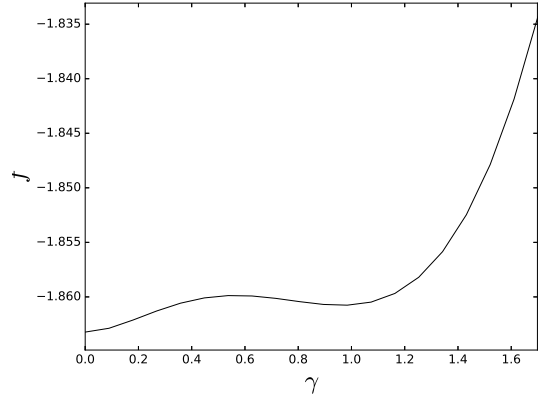}
        \captionof*{figure}{{\bf (b)} $\lambda = 0.91$, impossible}
    \end{subfigure}
    \begin{subfigure}[t]{0.47\textwidth}
        \centering
        \includegraphics[width=\linewidth]{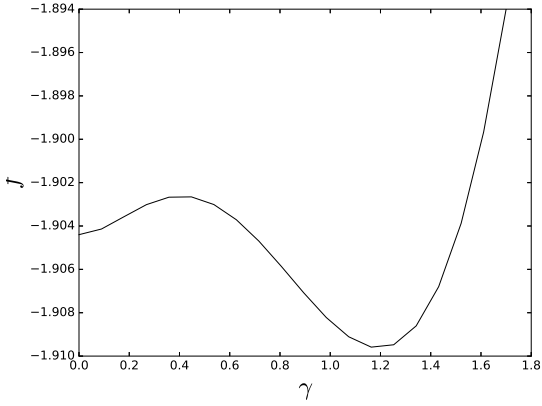}
        \captionof*{figure}{{\bf (c)} $\lambda = 0.92$, hard}
    \end{subfigure}
    \hfill
    \begin{subfigure}[t]{0.47\textwidth}
        \centering
        \includegraphics[width=\linewidth]{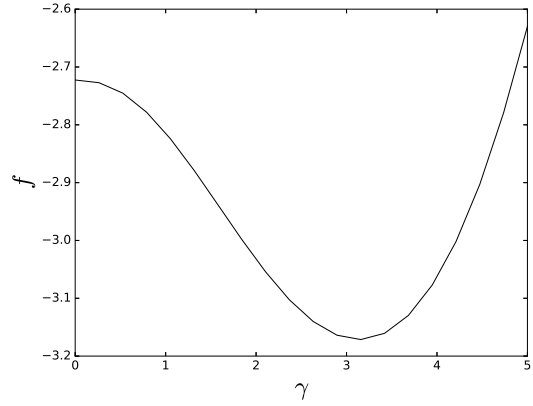}
        \captionof*{figure}{{\bf (d)} $\lambda = 1.1$, easy}
    \end{subfigure}
    \caption{Free energy landscape for $A_4$ with 1 frequency: the standard 3-dimensional representation (rigid motions of a tetrahedron). {\bf (a)} $\lambda = 0.8$. The global minimizer is $\gamma = 0$ so no estimator achieves nontrivial recovery. {\bf (b)} $\lambda = 9.1$. A new local minimum in the free energy has appeared, but the global minimum is still at $\gamma = 0$ and so nontrivial recovery remains impossible. {\bf (c)} $\lambda = 9.2$. AMP is stuck at $\gamma = 0$ but the (inefficient) statistically optimal estimator achieves a nontrivial $\gamma$ (the global minimum). AMP is not statistically optimal. This computational gap appears at $\lambda \approx 0.913$, at which point the global minimizer transitions discontinuously from $\gamma = 0$ to some positive value. {\bf (d)} $\lambda = 1.1$. AMP achieves optimal recovery. The AMP $\gamma$ value transitions discontinuously from zero to optimal at $\lambda = 1$.}
    \label{fig:a4}
\end{figure}

Next we move on to some 2-frequency problems, where $\gamma$ is now a 2-dimensional vector. In Figure~\ref{fig:2-freq} we see an example with no computational gap, and an example with a computational gap. Note that the free energy landscape at the AMP threshold $\lambda = (1,\ldots,1)$ reveals whether or not a computational gap exists: there is a gap if and only if the global minimum of free energy does not occur at the origin.

\begin{figure}[!ht]
    \centering
    \begin{subfigure}[t]{0.47\textwidth}
        \centering
        \includegraphics[width=\linewidth]{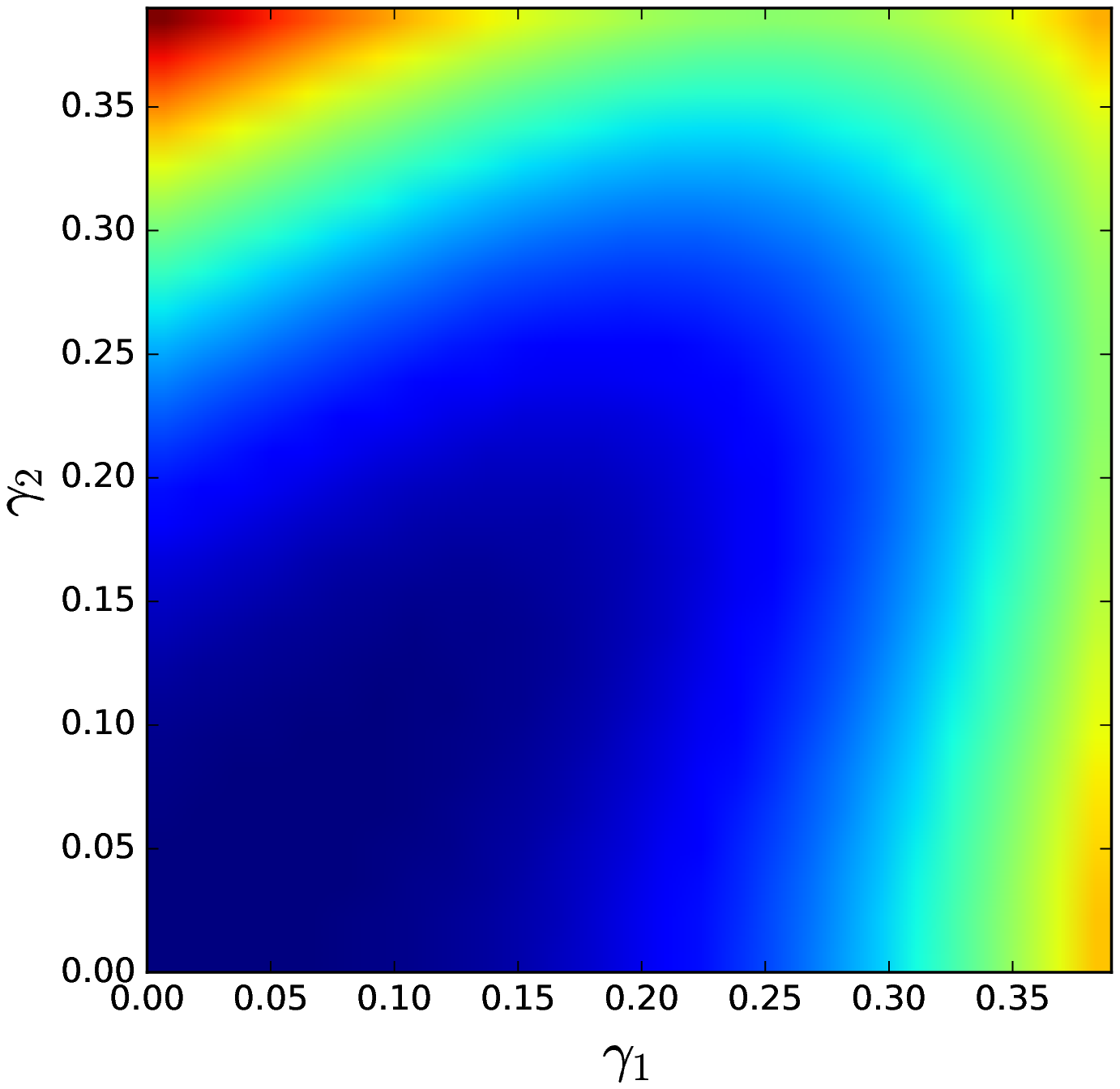}
    \end{subfigure}
    \hfill
    \begin{subfigure}[t]{0.47\textwidth}
        \centering
        \includegraphics[width=\linewidth]{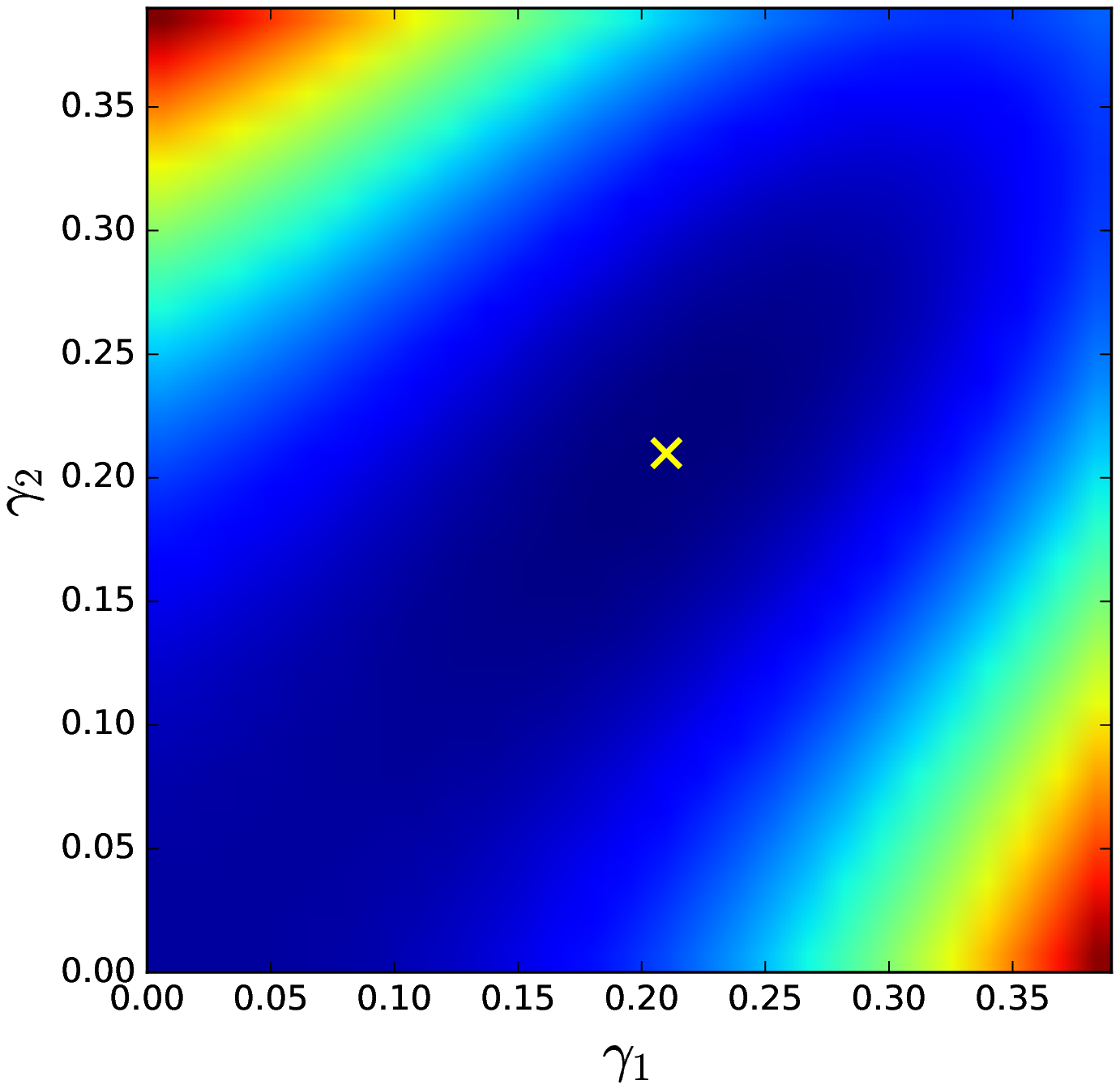}
    \end{subfigure}
    \caption{Free energy landscape for 2-frequency problems at the critical value $\lambda = (1,1)$. Darker colors indicate lower free energy. {\bf Left}: $\mathbb{Z}/6$ with 2 frequencies. Here the origin is the global minimizer of free energy and so there is no computational gap, i.e.\ nontrivial recovery is statistically impossible when both $\lambda_1$ and $\lambda_2$ are below 1. {\bf Right}: $\mathbb{Z}/5$ with 2 frequencies. Here the global minimizer (marked with an X) does not lie at the origin and so there is a computational gap, i.e. there is a regime where nontrivial recovery is statistically possible yet AMP fails.}
    \label{fig:2-freq}
\end{figure}

We now state some experimental results regarding which synchronization problems have computational gaps. For $U(1)$ with (the first) $K$ frequencies, there is a gap iff $K \ge 3$. For $\mathbb{Z}/L$ with $K$ frequencies, there is a gap for $K \ge 3$ and no gap for $K = 1$; when $K = 2$ there is only a gap for $L = 5$. For $SO(3)$ with $K$ frequencies, there is a gap iff $K \ge 2$.

In \cite{pwbm-contiguity} we gave some rigorous lower bounds for Gaussian synchronization problems, showing for instance that $U(1)$ with one frequency is statistically impossible below $\lambda = 1$. The non-rigorous results above predict further results that we were unable to show rigorously, e.g. $U(1)$ with two frequencies and $\mathbb{Z}/3$ (with one frequency) are statistically impossible below the $\lambda = 1$ threshold.

In the examples above we saw that when every $\lambda$ is below 1, AMP gives trivial performance, and when some $\lambda$ exceeds 1, AMP gives statistically optimal performance. However, the behavior can be more complicated, namely AMP can exhibit nontrivial but sub-optimal performance. In Figure~\ref{fig:subopt} we show such an example: $\mathbb{Z}/25$ with $9$ frequencies.

\begin{figure}[!ht]
    \centering
    \includegraphics[width=0.5\linewidth]{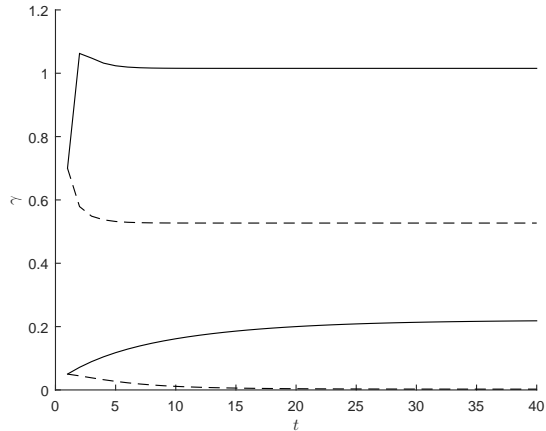}
    \caption{An example where AMP gives nontrivial but sub-optimal performance. Here we take $\mathbb{Z}/25$ with $9$ frequencies. Set $\lambda_k = 0.8$ for $k = 1,\ldots,8$ and $\lambda_9 = 1.1$. Since we cannot visualize the free energy landscape in 9 dimensions, we instead plot the state evolution recurrence as it evolves over time (number of iterations $t$) from two different starting points. The bottom two curves correspond to AMP's performance, where we initialize $\gamma$ to be small: $\gamma = (0.05,0.05)$. The solid line is $\gamma_9$ and the dashed line is $\gamma_1$ (which is representative of $\gamma_2,\ldots,\gamma_8$). The top two curves correspond to a ``warm start'' $\gamma = (0.7,0.7)$. We see that with the warm start, state evolution converges to a different fixed point with larger $\gamma$ values, and thus better correlation with the truth. Furthermore, this fixed point has lower free energy (not shown) than the lower one, indicating that the information-theoretically optimal estimator outperforms AMP.}
    \label{fig:subopt}
\end{figure}

\FloatBarrier

\section*{Acknowledgements}
The authors would like to thank Amit Singer, Roy Lederman, Yutong Chen, Nicholas Boumal, others from Amit Singer's group, and Yash Deshpande, for helpful discussions.

\bibliographystyle{alpha}
\bibliography{main}

\appendix

\section{Log-likelihood expansion for the Gaussian observation model}
\label{app:gaussian-loglikelihood}

In this section we show how the Gaussian observation model fits into the graphical model formulation by deriving the corresponding coefficient matrices $Y_\rho$. In particular, we show that $Y_\rho = d_\rho \lambda_\rho M_\rho$, a scalar multiple of the observed Gaussian matrix.

We can write $\log \cL_{uv}(g_u,g_v) = \sum_\rho \log \cL_{uv}^\rho(g_u,g_v)$ and consider each representation separately. There are three cases for the three types of representations (see Section~\ref{sec:types}).

For convenience we recall the Gaussian observation model:
$$M_\rho = \frac{{\lambda_\rho}}{n} X_\rho X_\rho^* + \frac{1}{\sqrt{n d_\rho}} W_\rho.$$
Restricting to the $u,v$ submatrix:
$$M_{uv}^\rho = \frac{\lambda_\rho}{n} \rho(g_u g_v^{-1}) + \frac{1}{\sqrt{nd_\rho}} W^\rho_{uv}.$$

\paragraph{Real type.}
Let $\rho$ be of real type. Recall that in this case, each entry of $W_{uv}^\rho$ is $\mathcal{N}(0,1)$. We have
\begin{align*}
\log \cL_{uv}^\rho(g_u,g_v)
&= \frac{-n d_\rho}{2} \left\| M_{uv}^{\rho} - \frac{\lambda_\rho}{n} \rho(g_u g_v^{-1}) \right\|_F^2 \\
&= \left\langle d_\rho \lambda_\rho M_{uv}^{\rho},\, \rho(g_u g_v^{-1}) \right\rangle  + \mathrm{const}.
\end{align*}
Here $\|\cdot\|_F$ denotes the Frobenius norm. The additive constant in the last step depends on $M_{uv}^\rho$ but not on $g_u,g_v$. Thus the log-likelihood coefficients are $Y_{uv}^\rho = d_\rho \lambda_\rho M_{uv}^\rho$ and so $Y_\rho = d_\rho \lambda_\rho M_\rho$.

\paragraph{Complex type.}
Now consider a representation $\rho$ of complex type, along with its conjugate $\bar{\rho}$. Recall that in this case, each entry of $W_{uv}^\rho$ has independent real and imaginary parts drawn from $\mathcal{N}(0,1/2)$. We have
\begin{align*}
\log \cL_{uv}^\rho(g_u,g_v)
&= -n d_\rho \left\| M_{uv}^{\rho} - \frac{\lambda_\rho}{n} \rho(g_u g_v^{-1}) \right\|_F^2 \\
&= \left\langle d_\rho \lambda_\rho M_{uv}^{\rho},\, \rho(g_u g_v^{-1}) \right\rangle + \left\langle d_\rho \lambda_\rho \bar{M_{uv}^{\rho}},\, \bar{\rho(g_u g_v^{-1})} \right\rangle  + \mathrm{const}.
\end{align*}
Therefore we have $Y_\rho = d_\rho \lambda_\rho M_\rho$ and $Y_{\bar{\rho}} = d_\rho \lambda_\rho \bar{M_{\rho}} = d_{\bar\rho} \lambda_{\bar\rho} {M_{\bar\rho}}$.

\paragraph{Quaternionic type.}
Now consider a representation $\rho$ of quaternionic type. Recall that in this case, $W_{uv}^\rho$ is block-quaternion where each $2 \times 2$ block encodes a quaternion value whose 4 entries are drawn independently from $\mathcal{N}(0,1/2)$. Note the following relation between the norm of a quaternion and its corresponding $2 \times 2$ matrix:
$$\|a+bi+cj+dk\|^2 \equiv a^2 + b^2 + c^2 + d^2 = \frac{1}{2} \left\|\begin{array}{cc} a+bi & c+di \\ -c+di & a-bi \end{array}\right\|_F^2.$$
We have
\begin{align*}
\log \cL_{uv}^\rho(g_u,g_v)
&= - n d_\rho \cdot \frac{1}{2} \left\| M_{uv}^{\rho} - \frac{\lambda_\rho}{n} \rho(g_u g_v^{-1}) \right\|_F^2 \\
&= d_\rho \lambda_\rho \,\mathfrak{Re}\left(\left\langle M_{uv}^{\rho}, \rho(g_u g_v^{-1}) \right\rangle\right) + \mathrm{const} \\
&= d_\rho \lambda_\rho \left\langle M_{uv}^{\rho}, \rho(g_u g_v^{-1}) \right\rangle + \mathrm{const}
\end{align*}
where $\mathfrak{Re}$ denotes real part. In the last step we used the fact that $M_{uv}^{\rho}$ and $\rho(g_u g_v^{-1})$ are block-quaternion and so their inner product is real (see Section~\ref{sec:types}). Therefore $Y_\rho = d_\rho \lambda_\rho M_\rho$.

\section{Proof of Lemma~\ref{lemma:A}}
\label{app:A}

To see that (i) and (ii) are equal, recall the interpretation of $\mathcal{F}_\rho$ as a conditional expectation: $\mathcal{F}_\rho(\cdots) = \EE[\rho(g) | \cdots]$ where $\cdots$ stands for $\left\{\gamma_q^t q(g) + \sqrt{\gamma_q^t} z_q\right\}_q$. (This is related to the \emph{Nishimori identities} in statistical physics.)

We have the following symmetry properties of $\mathcal{F}_\rho$.
\begin{lemma}
\label{lemma:sym}
\begin{enumerate}[(1)]
\item For any $\gamma_q^t \in \mathbb{R}$, $z_q \in \mathbb{C}^{d_\rho \times d_\rho}$, and $g,h \in G$, we have
$$\mathcal{F}_\rho\left(\left\{\gamma_q^t q(hg) + \sqrt{\gamma_q^t} z_q\right\}_q\right) = \rho(h) \mathcal{F}_\rho\left(\left\{\gamma_q^t q(g) + \sqrt{\gamma_q^t} q(h^{-1}) z_q\right\}_q\right)$$
and
$$\mathcal{F}_\rho\left(\left\{\gamma_q^t q(gh) + \sqrt{\gamma_q^t} z_q\right\}_q\right) = \mathcal{F}_\rho\left(\left\{\gamma_q^t q(g) + \sqrt{\gamma_q^t} z_q q(h^{-1}) \right\}_q\right)\rho(h).$$
\item Therefore, if we define
$$f_\rho(g) \equiv \EE_{z_q} \mathcal{F}_\rho\left(\left\{\gamma_q^t q(g) + \sqrt{\gamma_q^t} z_q\right\}_q\right)$$
we have $f_\rho(hg) = \rho(h)f_\rho(g)$ and $f_\rho(gh) = f_\rho(g)\rho(h)$.
\end{enumerate}
\end{lemma}
\begin{proof}
Part (1) is a straightforward computation using the definition of $\mathcal{F}_\rho$. Part (2) follows from part (1) because $z_q$ has the same distribution as $q(h^{-1}) z_q$ and $z_q q(h^{-1})$.
\end{proof}

We now return to the proof of Lemma~\ref{lemma:A}. The equality of (i) and (iii) follows from part (2) of Lemma~\ref{lemma:sym}. The equality of (ii) and (iv) follows from part (1) of Lemma~\ref{lemma:sym}. Combining this with the equality of (i) and (ii) from above, we have now shown equality of (i),(ii),(iii),(iv). It remains to show that $A_\rho^t$ is a real multiple of the identity.

Letting $e \in G$ be the identity, we have
$$f_\rho(e)\rho(g) = f_\rho(eg) = f_\rho(ge) = \rho(g)f_\rho(e)$$
and so by Schur's lemma, this means $f_\rho(e)$ is a (possibly complex) multiple of the identity. But $f_\rho(e)$ is just (iii), so we are done. To see that the multiple $a_\rho^t$ is real, note that the trace of (ii) is real.

\end{document}